%% file: main.tex
\documentclass[11pt]{article}

\usepackage[letterpaper,margin=1in]{geometry}
\input{preamble.tex}

\begin{document}
\allowdisplaybreaks
\title{Switching Time Optimization for Binary Quantum Optimal Control}

\author{Xinyu Fei\thanks{Department of Industrial and Operations Engineering, University of Michigan at Ann Arbor, USA;}
    ~~~Lucas T. Brady\thanks{KBR/NASA Ames Quantum Artificial Intelligence Laboratory;}
     ~~~Jeffrey Larson\thanks{Mathematics and Computer Science Division, Argonne National Laboratory;}
    ~~~Sven Leyffer\footnotemark[3]
     ~~~Siqian Shen\thanks{Department of Industrial and Operations Engineering, University of Michigan, Ann Arbor, USA, email: {\tt siqian@umich.edu}.}
    }
\date{~}
\maketitle
\begin{abstract}
  Quantum optimal control is a technique for controlling the evolution of a quantum system and has been applied to a wide range of problems in quantum physics.
    We study a binary quantum control optimization problem, where control decisions are binary-valued and the problem is solved in diverse quantum algorithms. 
  In this paper, we utilize classical optimization and computing techniques to develop an algorithmic framework that sequentially optimizes the number of control switches and the duration of each control interval on a continuous time horizon.
    Specifically, we first solve the continuous relaxation of the binary control problem based on time discretization and then use a heuristic to obtain a controller sequence with a penalty on the number of switches. 
    Then, we formulate a switching time optimization model and apply sequential least-squares programming with accelerated time-evolution simulation to solve the model. 
    We demonstrate that our computational framework can obtain binary controls with high-quality performance and also reduce computational time via solving a family of quantum control 
    instances in various quantum physics applications. 
\end{abstract}

\input{chpt1_introduction}
\input{chpt2_extraction_algorithm}
\input{chpt3_switching_optimization}
\input{chpt4_numeric}
\input{chpt5_conclusion}

~\\
~\\
\noindent {\bf Acknowledgement:}
    This work was supported in part by the U.S. Department of Energy, Office of Science, Office of Advanced Scientific Computing Research, Accelerated Research for Quantum Computing program under Contract No.~DE-AC02-06CH11357. X.~F.~and S.~S.~are partially supported by the Department of Energy (DoE) grant No.\ DE-SC0018018 and by the U.S.~National Science Foundation (NSF) Independent Research Development (IRD) program, respectively. 
L.~T.~B.~is a KBR employee working under the Prime Contract No.~80ARC020D0010 with the NASA Ames Research Center and is grateful for support from the DARPA Quantum Benchmarking program under IAA 8839, Annex 130.

\input{main.bbl}
%\bibliographystyle{unsrtnat}
%\bibliography{Xinyu}

\appendix
\input{appendix}
\end{document}

%% file: preamble.tex
\usepackage{subfig}% Support for small, `sub' figures and tables
\usepackage{amsfonts,amsmath,amssymb,amsthm}
\usepackage[boxruled, linesnumbered,vlined,algoruled]{algorithm2e}
\SetKwInput{KwInput}{Input}                % Set the Input
\SetKwInput{KwOutput}{Output} 
\SetKwProg{KwFn}{Function}{}{}
\SetKwComment{Comment}{/* }{ */}
\DontPrintSemicolon
\usepackage{hyperref}
\hypersetup{breaklinks=true, colorlinks=true, linkcolor={blue!80!black}, citecolor={blue!80!black}, urlcolor={blue!80!black}}
\usepackage{lscape}
\usepackage{graphicx}
\usepackage{tikz}
\usepackage{algorithmic}
\usetikzlibrary{trees}
\usepackage{pgfplots}
\usepackage{adjustbox}
\usepackage{multicol}
\usepackage{multirow}
\usepackage{mathtools}
\usepackage{float}
\usepackage[square, comma, sort&compress, numbers]{natbib}
\usepackage{bm}
\usepackage{comment}

\usetikzlibrary{shapes,positioning,fit,calc}

\def\argmax{\mathop{\rm arg\,max}}%
\def\argmin{\mathop{\rm arg\,min}}%

\pgfplotsset{compat=1.13}

\usepackage{colortbl}
\usepackage{pdflscape}
\usepackage{xcolor}
\usepackage{cite}               % ... more powerful citations package
\usepackage{bbold}

\newcolumntype{R}{>{\columncolor{red!20}}c}
\newcolumntype{G}{>{\columncolor{green!20}}c}
\newcolumntype{B}{>{\columncolor{blue!20}}c}
\newcolumntype{Y}{>{\columncolor{yellow!20}}c}
\newcolumntype{K}{>{\columncolor{black!20}}c}

\usepackage{courier}
\usepackage{listings}
\usepackage[toc,page]{appendix}

\newtheorem{theorem}{Theorem}
\newtheorem{lemma}{Lemma}
\newtheorem{proposition}{Proposition}
\newtheorem{corollary}{Corollary}

\newtheorem{remark}{Remark}

\graphicspath{{../figures/}}

\lstdefinestyle{C++}{
	language=C++,
	keywordstyle=\bfseries\color{purple},
	stringstyle=\color{blue},
	commentstyle=\color{gray},
	comment=[l]{/*}
}

\lstdefinestyle{AMPL}{
	language=AMPL,
	aboveskip=3mm,
	belowskip=3mm,
	showstringspaces=false,
	columns=flexible,
	keywordstyle=\bfseries,
	breaklines=true,
	breakatwhitespace=true,
	tabsize=3,
}

\lstdefinelanguage{AMPL}{keywords={let,set,param,var,arc,integer,minimize,maximize,subject,to,node,
sum,in,Current,complements,integer,solve_result_num,IN,contains,less,suffix,INOUT,default,logical,
Infinity,dimen,max,symbolic,Initial,div,min,table,LOCAL,else,option,then,OUT,environ,setof,union,
all,exists,shell_exitcodeuntil,binary,forall,solve_exitcodewhile,by,if,solve_messagewithin,check,
solve_result},sensitive=true,comment=[l]{\#}}

\newcommand{\be}{\begin{equation}}
\newcommand{\ee}{\end{equation}}
\newcommand{\bea}{\begin{eqnarray}}
\newcommand{\eea}{\end{eqnarray}}

\newcommand{\bvec}{\left(\begin{array}{c}}
\newcommand{\evec}{\end{array}\right)}
\newcommand{\bsub}{\begin{subequations}}
\newcommand{\esub}{\end{subequations}}

\usepackage{longtable}

\SetKwFunction{FRecurs}{FnRecursive}%
\SetKwProg{Function}{}{}{}\SetKwFunction{FRecurs}{void FnRecursive}%
\SetKwComment{Comment}{/* }{ */}

\usepackage{environ}

\definecolor{peach}{HTML}{FFCCAC}
\definecolor{butter}{HTML}{FFEB94}
\definecolor{babyblue}{HTML}{C1E1DC}

\NewEnviron{globalization_mechanism}{%
	\fcolorbox{peach}{peach}{%
	\begin{minipage}[l]{0.88\columnwidth}
	\BODY
	\end{minipage}
	}
}

\NewEnviron{globalization_strategy}{%
	\fcolorbox{butter}{butter}{%
	\begin{minipage}[l]{0.9\columnwidth}
	\BODY
	\end{minipage}
	}
}
\NewEnviron{step_computation}{%
	\fcolorbox{babyblue}{babyblue}{%
	\begin{minipage}[l]{0.9\columnwidth}
	\BODY
	\end{minipage}
	}
}

%% file: chpt1_introduction.tex
\section{Introduction}
\label{sec:intro}
{We use methods from classical computer design and engineering to accelerate the development of practical and scalable quantum computing systems.} Quantum control theory has a rich history that parallels the development of controllable quantum devices.  
Notions of quantum control theory \citep{Glaser_2015,Werschnik_2007,d2021introduction} first appeared in more analog settings, such as quantum chemistry \citep{Kosloff_1989,Peirce_1988}.  
As the field evolves, it saw applications in quantum information through gate design \citep{pawela2016various,Motzoi2009} and quantum circuit compilation \citep{gokhale2019partial}, and later has been used in more high-level design of quantum algorithms, especially with the advent of variational quantum circuits~\citep{Yang_2017,Bapat2019,Mbeng_2019,Lin_2019,brady2020optimal,Brady_2021b,Larocca_2022,Venuti_2021,Ge_2022,Fei2023binarycontrolpulse}.

To better implement quantum control, unlike most of the literature that solves quantum control problems using a fixed time discretization, we optimize both control functions and the time between consecutive control switches. 
{Our new algorithmic framework---based on classical computer and optimization models and algorithms---improves the quality of quantum controls and reduces computational time.}

Here we generalize the quantum control problem studied in our prior work (see~\citet{Fei2023binarycontrolpulse}) with binary variables and linear constraints, describing a so-called bang-bang control that is commonly used in variational quantum algorithms~\citep{farhi2014quantum,Bharti2022,Cerezo2021}. 
The quantum control problem is defined in a quantum system with $q$ qubits, an intrinsic Hamiltonian $H^{(0)}$, and $N$ control Hamiltonians $H^{(j)},\ j=1,\ldots,N$ with an evolution time horizon $[0, t_f]$. 
Matrices $X_\textrm{init}$ and $X_\textrm{targ}$ represent the initial and target unitary operator, respectively. 
For any time $t\in [0,t_f]$, we define control functions $u_j(t)$ for each controller $ j=1,\ldots,N$, a time-dependent Hamiltonian function $H(t)$, and a unitary operator function $X(t)$. 
Notation $u(t)=(u_1(t),\ldots,u_N(t))$ represents the vector form of control functions and $\mathcal{U}\subseteq \{0,1\}^N$ represents the feasible set of $u(t)$. 
A generic binary quantum control model is given by ~\citep{Fei2023binarycontrolpulse}: 
\begin{subequations}
\makeatletter
\def\@currentlabel{P}
\makeatother
\label{eq:model-c-1}
\begin{align}
    \label{eq:model-c-1-obj}
    (P) \quad \min_{u, X, H} \quad & F(X(t_f))\\
    \label{eq:model-c-1-cons-h}
    \textrm{s.t.}\quad & H(t) = H^{(0)} + \sum_{j=1}^N u_j(t)H^{(j)},\ \forall t\in [0,t_f]\\
    \label{eq:model-c-1-cons-s}
    & \frac{d}{dt} X(t) = -iH(t) X(t),\ \forall t\in [0,t_f]\\
    \label{eq:model-c-1-cons-i}
    & X(0) = X_\textrm{init}\\
    \label{eq:model-c-1-cons-u-def}
    & u(t)\in \mathcal{U},\ a.e.\ t\in [0,t_f].
\end{align}
\end{subequations}
The objective function~\eqref{eq:model-c-1-obj} is a general cost function only corresponding to the final unitary operator $X(t_f)$. We will introduce the specific functions discussed in our paper later. 
Constraint~\eqref{eq:model-c-1-cons-h} formulates the time-dependent Hamiltonian function $H(t)$ as a linear combination of the intrinsic Hamiltonian and the control Hamiltonians weighted by control functions. 
Constraint~\eqref{eq:model-c-1-cons-s} is the differential equation describing the time-evolution process of the unitary operator function $X(t)$ in our considered quantum systems with the same units for energy and frequency, which means that $\hbar =1$.
Constraint~\eqref{eq:model-c-1-cons-i} is the initial condition constraint of $X(t)$. 
Constraint~\eqref{eq:model-c-1-cons-u-def} indicates that for $t\in [0,t_f]$ almost everywhere, the feasible region set $\mathcal{U}$ constrains the value of the control function $u(t)$. In this paper, we consider two cases, one is binary controls $\mathcal{U} = \{0,1\}^N$, and the other one is binary controls with an additional constraint $\mathcal{U} = \{u\in \{0,1\}^N: \sum_{j=1}^N u_j=1\}$, 
where the constraint requires that the summation of all the control values should be one at any time, known as the Special Ordered Set of Type 1 (SOS1) property in discrete optimization~\citep{wolsey1999integer,sager2012integer}. For binary control functions, it is equivalent to requiring only one controller to be active at any time.

We employ two widely used objective functions in the quantum control field. We use $|\cdot\rangle$ to represent a quantum state vector and $\langle \cdot |$ to represent its conjugate transpose, and $\cdot^\dagger$ to represent the conjugate transpose of a complex matrix. One function is the energy ratio minimization function with the specific form
\begin{align}
\label{eq:obj-energy}
    F(X(t_f)) = 1 - \langle \psi_0 | X^\dagger(t_f) \tilde{H} X(t_f) |\psi_0\rangle / E_{\textrm{min}}, 
\end{align}
to minimize the energy corresponding to the Hamiltonian $\tilde{H}$. 
In the objective function, the initial state of the quantum system is given by $|\psi_0\rangle$ 
and the constant minimum energy $E_\textrm{min}$ represents the minimum eigenvalue of the Hamiltonian $\tilde{H}$. 
An alternative objective function is the infidelity function
\begin{align}
\label{eq:obj-fid}
    F(X(t_f)) = 1 - \frac{1}{2^q} \left|\textbf{tr} \left\{ X^\dagger_\textrm{targ} X(t_f)\right\}\right|, 
\end{align}
to minimize the difference between $X(t_f)$ and the target operator $X_\textrm{targ}$. Both objective functions are bounded between $[0,1]$.

In the literature, the most widely used method to solve the original quantum control problem~\eqref{eq:model-c-1} is time discretization, which divides the evolution time horizon into time intervals. We review the general discretized model proposed in ~\citet{Fei2023binarycontrolpulse} below. 
We divide the evolution time horizon $[0,t_f]$ into $T$ time intervals as $0=t_0<t_1<\ldots<t_T=t_f$. We use $k=1,\ldots,T$ to represent each time step with a corresponding time interval $[t_{k-1}, t_k)$.
We define a piece-wise constant control function based on time discretization. 
 For each time step $k=1,\ldots,T$ 
and each controller $j=1,\ldots,N$, we define control variables $u_{jk}$ with vector form $u_k=(u_{1k},\ldots,u_{Nk})^T$. 
For simplicity, we use $u$ to represent the control variables $u_{jk},\ j=1,\ldots,N,\ k=1,\ldots,T$.
Furthermore, we define time-dependent Hamiltonian variables $H_{k}$ and unitary operator variables $X_{k}=X(t_k)$ 
 for $k=1,\ldots,T$. The differential equation~\eqref{eq:model-c-1-cons-s} has the following explicit solution on the time interval $[t_{k-1}, t_k]$:
\begin{align}
    X_k=\exp\left\{-iH_k(t_k-t_{k-1})\right\}X_{k-1},\ k=1,\ldots,T,
\end{align}
where $H_k$ is defined in~\eqref{eq:model-d-1-cons-h}. 
With this explicit solution, the discretized quantum control problem has the following formulation:
\begin{subequations}
\makeatletter
\def\@currentlabel{$P_D$}
\makeatother
\label{eq:model-d-1} 
\begin{align}
    \label{eq:model-d-1-obj}
    (P_D)\quad \min_{u,X,H,t_1,\ldots,t_T} \quad & F(X_T)\\
    \label{eq:model-d-1-cons-h}
    \textrm{s.t.}\quad & H_k = H^{(0)} + \sum_{j=1}^N u_{jk}H^{(j)},\ k=1,\ldots,T\\
    \label{eq:model-d-1-cons-s}
    & X_{k}=e^{-i H_k (t_k-t_{k-1})}X_{k-1},\ k=1,\ldots,T \\
    \label{eq:model-d-1-cons-i}
    & X_0 = X_\textrm{init}\\
    \label{eq:model-d-1-cons-u-def}
    & u_k\in \mathcal{U},\ k=1,\ldots,T.
\end{align}
\end{subequations}
The objective function~\eqref{eq:model-d-1-obj} is the objective function~\eqref{eq:model-c-1-obj} with the approximated final operator $X_T\approx X(t_f)$. 
Constraints~\eqref{eq:model-d-1-cons-h}--\eqref{eq:model-d-1-cons-u-def} are the time-discretized versions of constraints~\eqref{eq:model-c-1-cons-h}--\eqref{eq:model-c-1-cons-u-def}.
The simultaneous optimization of both control functions $u,\ X,\ H$, and time discretization points $t_1,\ldots,t_T$ was proposed by~\citet{logsdon1989accurate} for differential equation systems. 
Due to the computational complexity, we divide our problem into two parts: (i) optimizing discretized control variables $u_k$ for each time step $k=1,\ldots,T$ and (ii) optimizing time points of switching between controllers.
Next, we review the main literature focusing on the problem of each part.

\paragraph{Literature of Quantum Optimal Control.} Most quantum optimal control algorithms only consider optimizing controls under a {\it fixed} time discretization. 
For continuous quantum control problems, \citet{khaneja2005optimal} first proposed the gradient ascent pulse engineering (GRAPE) algorithm to estimate control functions with piece-wise constant functions and apply gradient-based methods to optimize them. 
Furthermore, \citet{Larocca_2021} combined the GRAPE algorithm and Krylov-subspace approach to improve computational efficiency. 
\citet{brady2020optimal} developed an analytical framework based on time discretization for optimizing bang-bang and smooth annealing controls for a quantum energy minimization problem. 
Other algorithms based on time discretization include pseudospectral methods~\citep{li2009pseudospectral} and reinforcement learning frameworks~\citep{bukov2018reinforcement,niu2019universal,sivak2022model}. 
For binary quantum control problems, \citet{vogt2022binary} applied a trust-region method~\citep{nocedal2006numerical} to optimize controls in a single flux quantum system. 
\citet{Fei2023binarycontrolpulse} developed a solution framework combining the GRAPE algorithm, rounding techniques, and local branching heuristics. 
However, obtaining controls with higher quality by the discretized model requires more precise time discretization, equivalent to a larger number of decision variables associated with each time step, which increases the computational cost severely. 
Estimating control functions by piece-wise functions also leads to limitations in obtaining better control solutions.

\paragraph{Literature of Switching Time Optimization.} The switching time optimization has been a meaningful yet challenging topic in controls of switched-mode dynamic systems~\citep{meier1990efficient,egerstedt2006transition,johnson2011second,flasskamp2013discretized,stellato2017second}. 
\citet{vossen2010switching} discussed the switching time optimization for both bang-bang and singular controls. 
In quantum theory, a class of literature formulates the quantum approximate optimization algorithm (QAOA) as an optimal control problem with two controllers and applies multiple methods to solve it~\citep{farhi2014quantum,liang2020investigating,bao2018optimal}, which can be considered as a simplified version of switching time optimization with only two switching controllers. 
The aforementioned literature assumes known analytical formulations of control functions. \citet{bukov2018reinforcement} designed simple variational control protocols from discretized control solutions obtained by reinforcement learning which requires numerous training episodes and lacks generality for various control problems. 

\paragraph{Main Contributions.} The main contributions of this paper are as follows. First, we develop a new algorithmic framework that not only optimize control functions but also switching time points. 
Specifically, we develop two heuristic methods to obtain binary discretized controls from continuous discretized controls balancing the quality of controls and the number of switches for general quantum control problems. 
Second, we build and solve a generic switching time optimization model for quantum control problems with given Hamiltonian controllers as parameters. 
In addition, we accelerate the time-evolution simulations by pre-computing the eigenvalues of a small number of Hamiltonian matrices.
Third, we conduct numerical simulations on multiple quantum control examples and show that our method obtains controls with higher quality and a similar number of switches within significantly less computational time compared to the discretized model.  

The remainder of the paper is organized as follows. 
In Section~\ref{sec:alg-extraction}, we construct our algorithmic framework and develop heuristic methods to obtain controller sequences. 
In Section~\ref{sec:model-switching}, we formulate and solve the switching time optimization model with acceleration techniques for time evolution. 
In Section~\ref{sec:numerical}, we conduct numerical simulations and discuss the results. 
In Section~\ref{sec:conclusion}, we conclude our work and propose future research directions.

%% file: chpt2_extraction_algorithm.tex
\section{Algorithm Framework and Controller Sequence Extraction}
\label{sec:alg-extraction}
Our algorithmic framework consists of four steps, which we describe in detail in Algorithm~\ref{alg:overall}. First, we solve the continuous relaxation of the discretized model~\eqref{eq:model-d-1} with a fixed (equal) time discretization and time interval length $\Delta t$ by the penalized GRAPE algorithm proposed in the paper by~\citet{Fei2023binarycontrolpulse}. 
Second, we obtain binary controls from solutions of continuous relaxation by rounding algorithms, such as algorithms in~\citep[Table 1]{Fei2023binarycontrolpulse}, Algorithm~\ref{alg:switch-eobj}, and Algorithm~\ref{alg:switch-cumdiff}.
Third, we merge consecutive time intervals with the same control values as a time interval to derive a sequence of controllers by Algorithm~\ref{alg:extract-switch}. 
Forth, we optimize the switching time points of merged time intervals with the sequence of controllers as parameters, which is discussed in detail in Section~\ref{sec:model-switching}. We present the flow chart for our overall algorithm in Figure~\ref{fig:flowchart}.
\begin{algorithm}[!ht] \caption{Switching time optimization method for binary quantum control problem.\label{alg:overall}} 
    \DontPrintSemicolon
    \SetNoFillComment
    \KwInput{Discretized binary quantum control model~\eqref{eq:model-d-1}.}
    Solve the continuous relaxation of~\eqref{eq:model-d-1} with a fixed time discretization and obtain continuous solutions $u^{\textrm{con}}$.\;
    \label{algline:solvec}
    Obtain heuristic binary solutions $u^\textrm{bin}$ from continuous solutions $u^\textrm{con}$ by a given rounding algorithm, such as algorithms in~\citep[Table 1]{Fei2023binarycontrolpulse}, Algorithms~\ref{alg:switch-eobj}, and~\ref{alg:switch-cumdiff}.\;
        \label{algline:round}
    Merge time intervals with the same control value and derive the controller sequence $\bar{\mathcal{H}}$ with length $S$ by Algorithm~\ref{alg:extract-switch}.\;
        \label{algline:merge}
    Solve switching time optimization model~\eqref{eq:model-ts} in Section~\ref{sec:model-switching} to optimize switching time points with the given controller sequence $\bar{\mathcal{H}}$.\;
        \label{algline:st} 
    \KwOutput{Control time interval lengths for each controller in the controller sequence $\bar{\mathcal{H}}$.}
\end{algorithm}
\begin{algorithm}[!ht]
\caption{Controller sequence extraction from binary controls. \label{alg:extract-switch}}
 \DontPrintSemicolon
    \SetNoFillComment
 \KwInput{Discretized binary controls $u^{\textrm{bin}}$.}
    Initialize controller sequence $\bar{\mathcal{H}}_1 = \left[H^{(0)} +\sum_{j=1}^N u_{j1}H^{(j)}\right]$ and sequence length $S\gets1$.\;
    \For{$k=2,\ldots,T$}{
    \If{$u_{k}\neq u_{k-1}$\label{algline:update-swith-start}} 
    {$S\leftarrow S + 1,\ \bar{H}_S = H^{(0)} + \sum_{j=1}^N u_{jk} H^{(j)}$.\;
    $\Bar{\mathcal{H}}_S\leftarrow \left[\bar{\mathcal{H}}_{S-1},\bar{H}_S\right]$.}
    }
    \KwOutput{Controller sequence $\Bar{\mathcal{H}}_S$ and length $S$. }
\end{algorithm}
\begin{figure}[ht]
    \centering
    \includegraphics[width=0.9\textwidth]{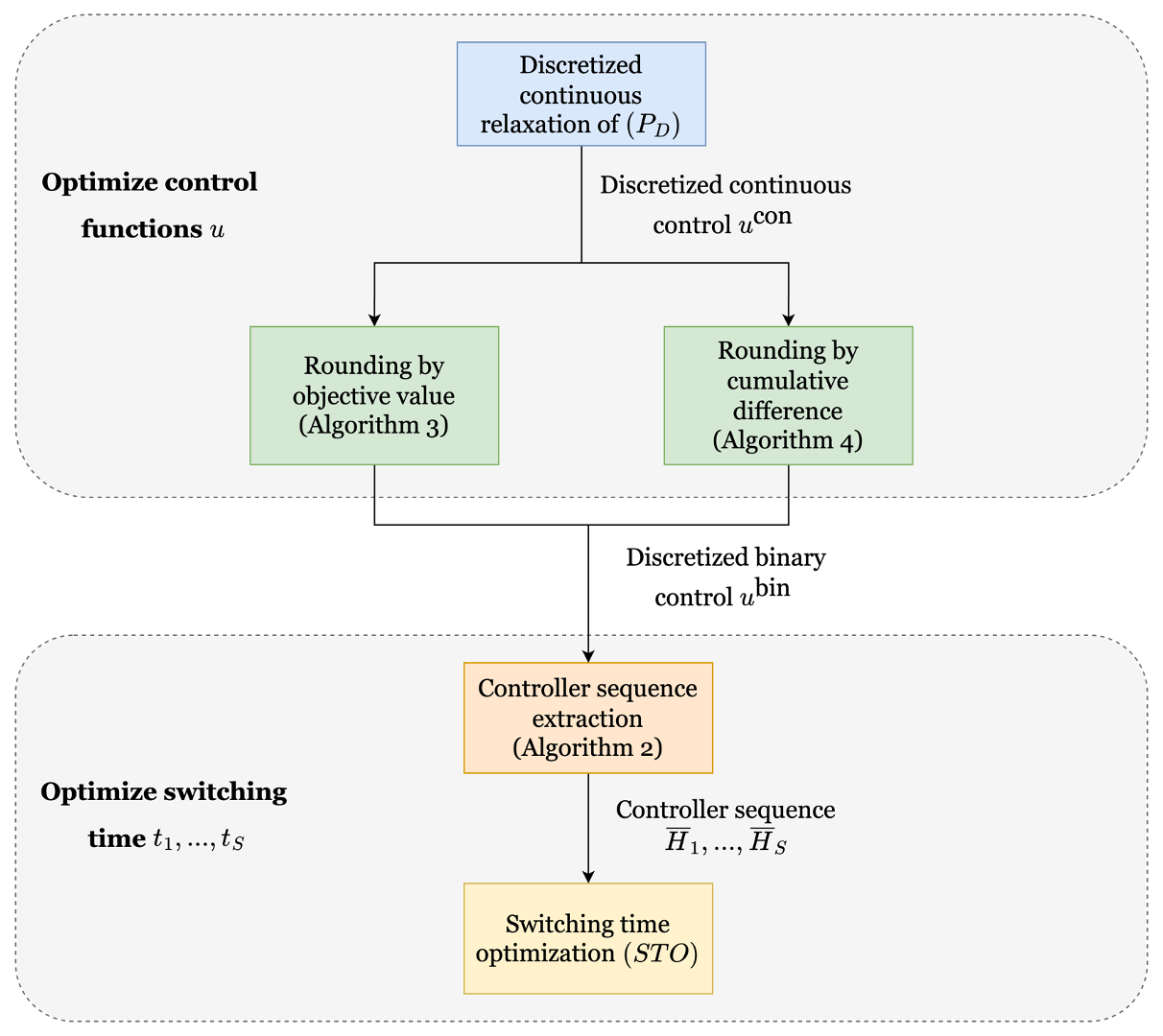}
\caption{Flow chart for the overall algorithmic framework (Algorithm~\ref{alg:overall}). We first optimize control functions $u$, including solving the discretized continuous relaxation of model~\eqref{eq:model-d-1} and applying rounding algorithms (Algorithm~\ref{alg:switch-eobj}--\ref{alg:switch-cumdiff}). Then we optimize the switching times by solving the switching time optimization model~\eqref{eq:model-ts} with controller sequences extracted from binary controls (Algorithm~\ref{alg:extract-switch}) as input.}
    \label{fig:flowchart}
\end{figure}

We present control results of a simple example in Figure~\ref{fig:example-overall}.
In the left figure, we present the piece-wise continuous control function $u^\textrm{con}$ obtained after solving the discretized continuous relaxation (step~\ref{algline:solvec}). 
In the middle figure, we present the binary control obtained after rounding methods (step~\ref{algline:round}). 
We convert the control functions from fractional values to binary values. 
In the right figure, we present the optimized control after optimizing the time of switches obtained by step~\ref{algline:st}.
\begin{figure}[htbp]
    \centering
\subfloat{\includegraphics[width=0.32\textwidth]{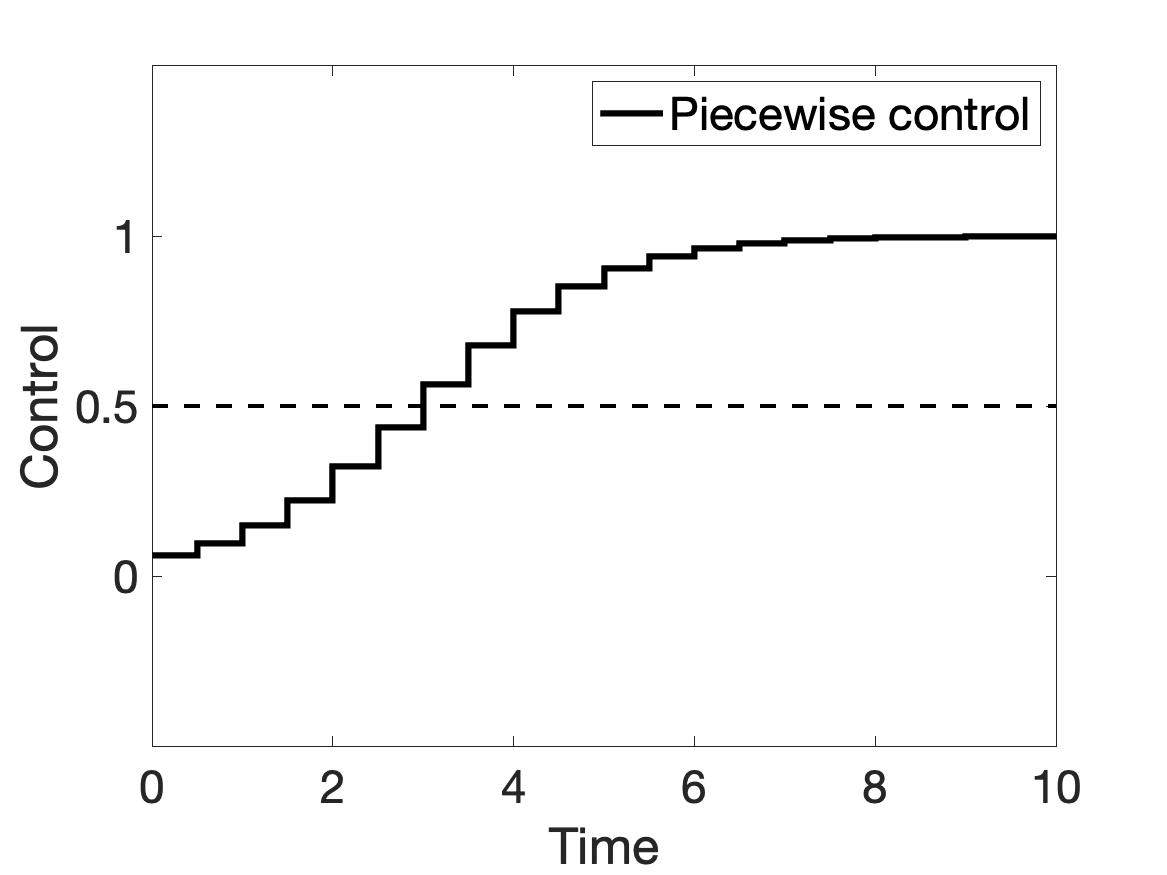}\label{fig:example-overall-c}}
\subfloat{\includegraphics[width=0.32\textwidth]{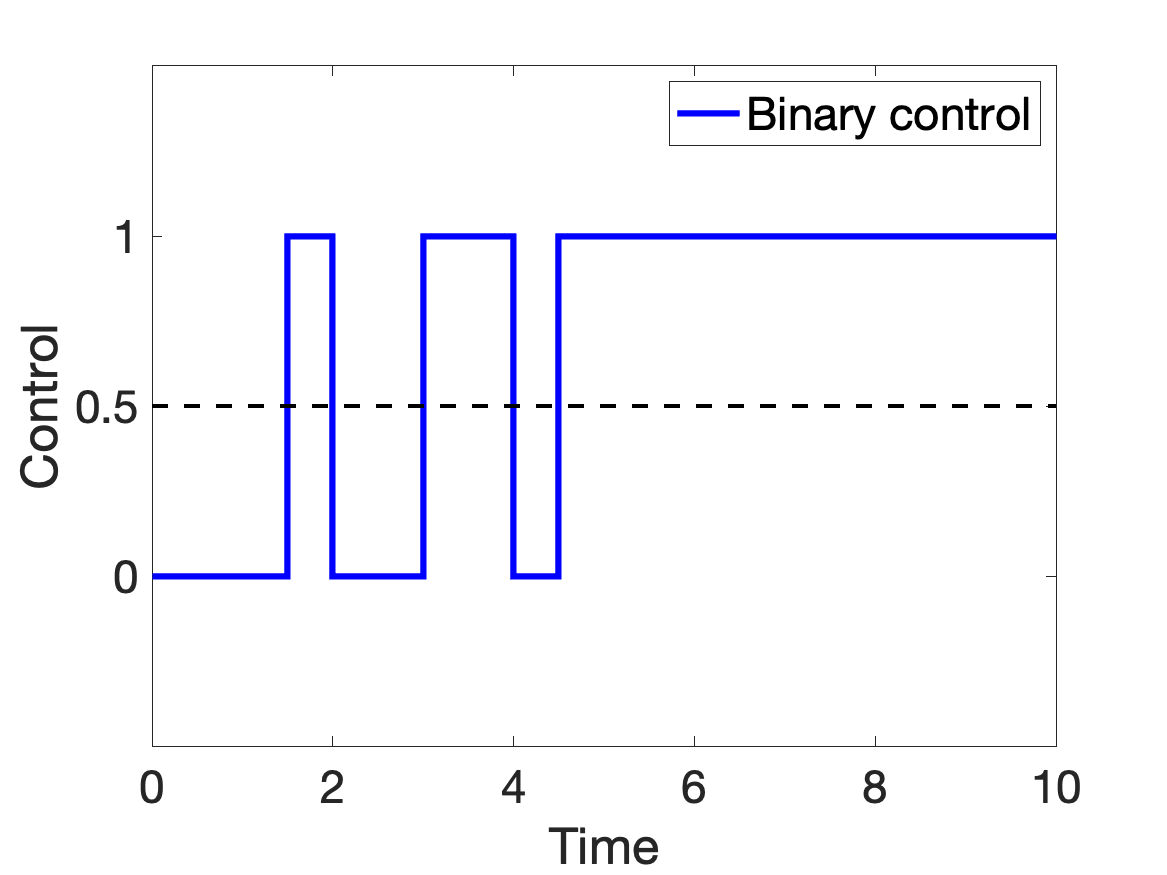}\label{fig:example-overall-b}}
\subfloat{\includegraphics[width=0.32\textwidth]{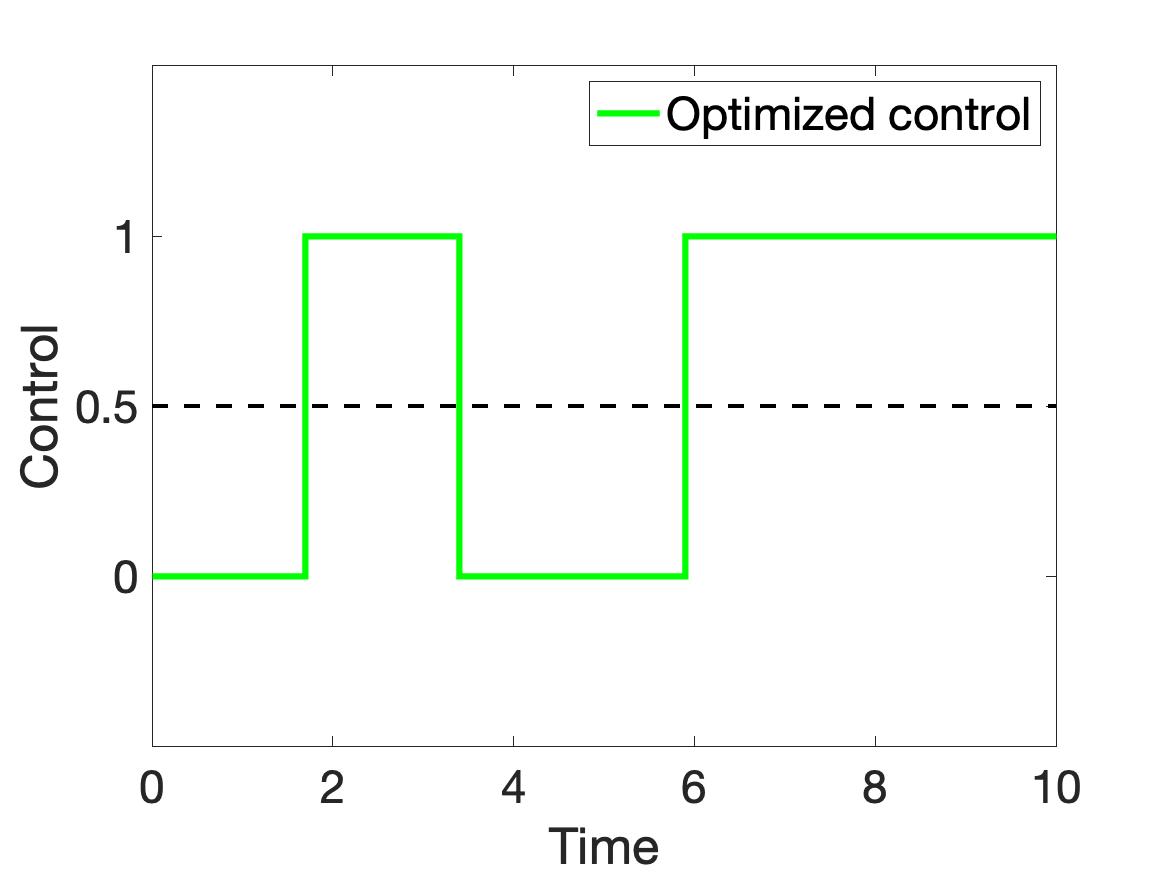}\label{fig:example-overall-st}}
    \caption{Simple example for the overall algorithm framework (Algorithm~\ref{alg:overall}). Left: Piece-wise continuous control $u^\textrm{con}$ after step~\ref{algline:solvec}. Middle: Binary control $u^\textrm{bin}$ after step~\ref{algline:round}. Right: Optimized control after step~\ref{algline:st}.}
    \label{fig:example-overall}
\end{figure}

Various methods to obtain binary controls have been proposed but they either require high computational costs or frequent switches~\citep{Fei2023binarycontrolpulse,sager2012integer,you2011mixed,manns2020multidimensional,sager2011combinatorial}. 
In this section, we present two heuristic rounding methods to obtain controller sequences from given continuous discretized controls $u^\textrm{con}$. 
Our key idea is to balance the difference between continuous and binary controls and the number of switches. For any control $u$, we evaluate the frequency of switching by the total variational (TV) norm defined as cumulative absolute differences between consecutive steps with the formulation:
\begin{align}
    TV(u)=\sum_{k=1}^{T-1} \sum_{j=1}^N \left|u_{jk} - u_{jk+1}\right|.
\end{align}
We evaluate the difference based on two metrics and propose the corresponding methods in Sections~\ref{sec:alg-obj} and \ref{sec:alg-sur}, respectively. 
In each section, we also prove that for both methods, if we only consider minimizing the difference between continuous and binary controls, when the length of time intervals goes to zero, the limit of the objective value of binary controls is no more than the limit of the objective value of the continuous control under the objective functions~\eqref{eq:obj-energy} and \eqref{eq:obj-fid} considered in this paper. 

\subsection{Method Based on Objective Value}
\label{sec:alg-obj} 
We convert unitary operator variables $X$ and Hamiltonian variables $H$ to implicit functions of control variables $u$ by constraints~\eqref{eq:model-d-1-cons-h}--\eqref{eq:model-d-1-cons-i}. The final operator $X_T$ can also be computed as a function of $u$ as
\begin{align}
    X_T(u) = \prod_{k=1}^T e^{-i\left(H^{(0)} + \sum_{j=1}^N u_{jk} H^{(j)}\right)\Delta
    t}X_{\textrm{init}}, 
\end{align}
which is equivalent to solving the Schr\"odinger equation~\eqref{eq:model-d-1-cons-h}--\eqref{eq:model-d-1-cons-i}. 
Substituting $X_T(u)$ into the objective function $F$, we derive an objective function only dependent on $u$, denoted as $\bar{F}(u)$ by eliminating intermediate variables $H_k$ and $X_k$, $k=1,\ldots,T $. 

In this section, we propose Algorithm~\ref{alg:switch-eobj} based on evaluating the difference between binary and continuous solutions by the difference of objective values $\bar{F}(u^\textrm{bin})$ and $\bar{F}(u^\textrm{con})$. We define $\alpha$ as the TV norm penalty parameter.
We present the explanation as follows. 
At each time step $k$, we examine the current control $[u_1^\textrm{bin},\ldots,u_{k-1}^\textrm{bin}, \hat{u}, u_{k+1}^\textrm{con},\ldots,u_{T}^\textrm{con}]\in [0,1]^{N\times T}$ for all $\hat{u}\in \mathcal{U}$ and choose the control with smallest objective value $\hat{u}^*$ (see step~\ref{algline:obj-ite}--\ref{algline:obj-evaluate}).  
We compare the choices of keeping the control values of the previous time step, i.e. $\hat{u} = u_{k-1}^\textrm{bin}$, and updating control values, i.e. $\hat{u} = \hat{u}^*$ (see step~\ref{algline:obj-condition}). 
If the objective value of keeping the control is no larger than the TV-norm value of updating the control weighted by $\alpha$, we keep the control (see step~\ref{algline:obj-keep}). 
Otherwise, we choose the control with the smallest objective value and update the TV-norm value (see step~\ref{algline:obj-update}). 
If $\alpha=0$, we always choose the control with the smallest objective value. 
Furthermore, we update the binary control $u^\textrm{bin}_k$ by the selected control and move to the next time step $k+1$ (see step~\ref{algline:obj-update-last}). This approach requires $T-1$ simulations at a receding time horizon. 

The computational advantages of conducting time-evolution processes on quantum computers allow computing objective values quickly. For simulations on classical computers, we propose an acceleration technique to avoid conducting time evolution for every examined control. 
With continuous solutions $u^\textrm{con}$, we define back propagators $\mu_k$ as 
\begin{align}
\label{eq:acc-eobj-back}
    \mu_k = \mu_{k+1}e^{-i\left(H^{(0)} + \sum_{j=1}^N u^\textrm{con}_{jk} H^{(j)}\right)\Delta t},\ k=1,\ldots,T,
\end{align}
where $\mu_{T+1}$ is an identity matrix. At each time step, after determining the binary control, the operator corresponding to binary controls is updated as
\begin{align}
\label{eq:acc-state}
    % X^b_{k} = e^{-i(H^{(0)} + H^{(\hat{j}(k))})\Delta t}X^b_{k-1}, 
    X^\textrm{bin}_{k} = e^{-i\left(H^{(0)} + \sum_{j=1}^N u^\textrm{bin}_{jk}H^{(j)}\right)\Delta t}X^\textrm{bin}_{k-1}, 
\end{align}
where $X_0^\textrm{bin}=X_\textrm{init}$. We have the following proposition to evaluate the objective value. 
\begin{proposition}
\label{prop:acc-obj}
    At time step $k=1,\ldots,T$, the objective value of a current examined control is computed by
    \begin{align}
\label{eq:acc-obj}
    \bar{F}\left([u_1^\textrm{bin},\ldots,u_{k-1}^\textrm{bin}, \hat{u}, u_{k+1}^\textrm{con},\ldots,u_{T}^\textrm{con}]\right) = F\left(\mu_{k+1}e^{-i\left(H^{(0)} + \sum_{j=1}^N \hat{u}_j H^{(j)}\right)\Delta t}X^\textrm{bin}_{k-1}\right),\ \forall \hat{u}\in \mathcal{U}.
\end{align}
\end{proposition}
\begin{proof}
    For any examined control, the final operator $\hat{X}_T$ is computed by 
    \begin{align}
        \displaystyle \hat{X}_T = \prod_{t=k+1}^{T} e^{-i\left(H^{(0)} + \sum_{j=1}^N u^\textrm{con}_{jt} H^{(j)}\right)\Delta t}X_k
        =\mu_{k+1} X_k = \mu_{k+1} e^{-i\left(H^{(0)} + \sum_{j=1}^N \hat{u}_j H^{(j)}\right)\Delta t} X_{k-1}^\textrm{bin}, 
    \end{align}
    where the last two equalities directly follow the definition $\mu_{k}$ and $X_{k}^\textrm{bin}$ in~\eqref{eq:acc-eobj-back}--\eqref{eq:acc-state}. 
    From the definition of $\Bar{F}$, we have
    \begin{align}
        \bar{F}\left([u_1^\textrm{bin},\ldots,u_{k-1}^\textrm{bin}, \hat{u}, u_{k+1}^\textrm{con},\ldots,u_{T}^\textrm{con}]\right) 
        = F(\hat{X}_T) = F\left(\mu_{k+1}e^{-i\left(H^{(0)} + \sum_{j=1}^N \hat{u}_j H^{(j)}\right)\Delta t}X^\textrm{bin}_{k-1}\right).
    \end{align}
\end{proof}
From Proposition~\ref{prop:acc-obj}, we show that the pre-computing and the reformulation of evaluating the objective value significantly reduce the matrix operation times, which we discuss in detail in the following proposition. 
\begin{proposition}
    Let $N'=\mathcal{U}$ be the number of feasible examined controls at each time step. With Algorithm~\ref{alg:switch-eobj}, we reduce the computation of matrix exponentials from $O(N'T^2)$ to $O(N'+T)$ and the computation of matrix multiplications from $O(N'T^2)$ to $O(N'T)$.
\end{proposition}
\begin{proof}
    Without our reformulation of the objective value and pre-computation, we require computing $O(T)$ matrix exponentials and $O(T)$ matrix multiplications during each time evolution process. We need to evaluate all the examined controls at each time step, which means that we need to do $N'T$ time evolution processes, hence the total number is $O(N'T^2)$ for both matrix exponentials and multiplications. 
    With our acceleration, we only need to compute $T$ unitary propagators of the continuous control and the matrix exponential of each possible control in $\mathcal{U}$. 
    Hence, the order of computing matrix exponentials is $O(N'+T)$. 
    We need to perform $T$ matrix multiplications for computing $\mu_k,\ k=1,\ldots,T$. 
    At each time step, we need to compute the objective value of setting each controller active, which requires total $O(N')$ matrix multiplications, and then update the current operator $X_k^\textrm{bin}$, which only requires $1$ multiplication. Therefore, taking a summation of all the time steps, the order of matrix multiplications is $O(N'T)$. 
\end{proof}
We present the details of the algorithm in Algorithm~\ref{alg:switch-eobj}. 
\begin{algorithm}[ht] \caption{Heuristic Rounding Methods based on Objective Values (Obj)
\def\@currentlabel{{\it Obj}}
\label{alg:switch-eobj}
} 
    \DontPrintSemicolon
    \SetNoFillComment
    \KwInput{Continuous control $u^\textrm{con}$ and constant TV norm parameter $\alpha$.}
    Initialize binary control $u^\textrm{bin}=u^\textrm{con}$ and current objective value $F_\textrm{cur}=\bar{F}(u^\textrm{con})$.\;
    Pre-compute back propagators $\mu_k,\ k=1,\ldots,T$ by Eqn~\eqref{eq:acc-eobj-back}.\;
    Pre-compute matrix exponentials $e^{-i\left(H^{(0)}+\sum_{j=1}^N \hat{u}_{j}H^{(j)}\right)}$ for all $\hat{u}\in \mathcal{U}$.\;
    Let $u_1^\mathrm{bin}=\argmin_{\hat{u}\in \mathcal{U}} \bar{F}\left([\hat{u}, u_2^\mathrm{con},\ldots,u_{T}^\mathrm{con}]\right)$.\;
    \For{$k=2,\ldots,T$
    \label{algline:obj-ite}}{
    Let $\hat{u}^* = \argmin_{\hat{u}\in \mathcal{U}} \bar{F}\left([u_1^\textrm{bin},\ldots,u_{k-1}^\textrm{bin}, \hat{u}, u_{k+1}^\textrm{con},\ldots,u_{T}^\textrm{con}]\right)$.\;
        \label{algline:obj-evaluate}
\eIf{$\bar{F}\left([u_1^\textrm{bin},\ldots,u_{k-1}^\textrm{bin}, u_{k-1}^\textrm{bin}, u_{k+1}^\textrm{con},\ldots,u_{T}^\textrm{con}]\right) \leq \alpha TV([u_1^\textrm{bin},\ldots,u_{k-1}^\textrm{bin}, \hat{u}^*, u_{k+1}^\textrm{con},\ldots,u_{T}^\textrm{con}])$\label{algline:obj-condition}}
    {
    Update binary control $u_k^\textrm{bin} = u_{k-1}^\textrm{bin}$.\;
    \label{algline:obj-keep}}   
    {
    Update binary control $u_k^\textrm{bin} = \hat{u}^*$, break ties 
    with smaller TV value.\;
    \label{algline:obj-update}}
    Update $X_k^\textrm{bin}$ by Eqn.~\eqref{eq:acc-state} and objective value $F_\textrm{cur}=\bar{F}\left([u_1^\textrm{bin},\ldots,u_{k}^\textrm{bin}, u_{k+1}^\textrm{con},\ldots,u_{T}^\textrm{con}]\right)$.\;
    \label{algline:obj-update-last}}
    \KwOutput{Binary control $u^\textrm{bin}$.}
\end{algorithm}
For any bounded objective function with upper bound $F_{ub}$, the TV-norm value with penalty parameter $\alpha$ is upper-bounded by $F_{ub}/\alpha$. In the following theorems, we discuss the difference between the objective values of binary controls $u^\textrm{bin}$ and input continuous controls $u^\textrm{con}$ when setting the penalty parameter $\alpha=0$. All the conclusions in this section are proved for both the energy objective function~\eqref{eq:obj-energy} and the infidelity objective function~\eqref{eq:obj-fid}. 
Before presenting our main results, we review a series of lemmas corresponding to unitary matrices and singular values used in our proof. The proof of Lemma~\ref{lemma:singular-1}--\ref{lemma:singular-3} is presented by~\citet{horn_johnson_1991}. 
\begin{lemma}
\label{lemma:singular-1}
(\citep[Section 3.1]{horn_johnson_1991})
The singular values are invariant under the multiplication of unitary matrices. 
\end{lemma}
\begin{lemma}
\label{lemma:singular-2}
(\citep[Theorem 3.3.4, 3.3.16]{horn_johnson_1991})
Let $A,\ B\in \mathbb{C}^{m\times m}$ be two complex matrices and let $\sigma_1(\cdot)$ represent the maximum singular value of a matrix, then
    $ \sigma_1(AB)\leq \sigma_1(A)\sigma_1(B),\ \sigma_1(A+B)\leq \sigma_1(A) + \sigma_1(B)$.
\end{lemma}
\begin{lemma}
\label{lemma:singular-3}
(\citep[Theorem 3.1.2]{horn_johnson_1991}) 
Let $A$ be an arbitrary complex matrix, then for any unit vector $|\psi\rangle$ such that $\||\psi\rangle\|_2=1$, it holds that $\left|\langle \psi|A|\psi\rangle\right| \leq \sigma_1(A)$.
\end{lemma}
\begin{lemma}
\label{lemma:singular-4}
Let $U\in \mathbb{C}^{m\times m}$ be a unitary matrix and $A\in \mathbb{C}^{m\times m}$ be an arbitrary complex matrix, then
    $\displaystyle \left|\mathbf{tr}\left\{UA\right\}\right|\leq \sum_{i=1}^m \sigma_i(A)\leq m\sigma_1(A)$, 
where $\sigma_1(A)\geq \ldots\geq \sigma_m(A)\geq 0$ are the singular values of $A$. 
\end{lemma}
The proof of Lemma~\ref{lemma:singular-4} is presented in Appendix~\ref{app:proof}. 
In the following theorem, we prove that if $u^\textrm{con}$ has the SOS1 property, for the energy objective function and infidelity objective function discussed in Section~\ref{sec:intro}, the difference of objective values $\bar{F}(u^\textrm{bin})-\bar{F}(u^\textrm{con})$ is bounded by $O(\Delta t)$. 
\begin{theorem}
\label{theo:str}
Let $u^\textrm{con}$ be the solution of the continuous relaxation of~\eqref{eq:model-d-1} with a fixed time discretization and $u^\textrm{bin}$ be the binary controls obtained by Algorithm~\ref{alg:switch-eobj} when $\alpha=0$. If the SOS1 property holds for continuous controls, which means that $\sum_{j=1}^N u^\textrm{con}_{jk} = 1,\ k=1,\ldots,T$, and is required for binary controls, then there exists constants $C_1,\ C_2 > 0$ such that for energy objective function~\eqref{eq:obj-energy} and infidelity objective function~\eqref{eq:obj-fid}, the difference between the objective value of binary and continuous control satisfies
\begin{align}
    \label{eq:diff}
    \bar{F}(u^\textrm{bin}) - \bar{F}(u^\textrm{con})\leq 2C_1e^{C_2\Delta t}\Delta t.
\end{align}
Furthermore, we have \begin{align}
    \limsup\limits_{\Delta t\rightarrow 0} \bar{F}(u^\textrm{bin}) - \bar{F}(u^\textrm{con})\leq 0.
\end{align}
\end{theorem}
The detailed proof is presented in Appendix~\ref{app:proof}. 
Next, we discuss a more general case when the SOS1 property of the continuous control $u^\textrm{con}$ does not exactly hold but the SOS1 property of the binary control $u^\textrm{bin}$ is still required. 
We define the cumulative error between binary and continuous control as
\begin{align}
    \epsilon^c(\Delta t) = \sum_{k=1}^{T} \left|\sum_{j=1}^N u^\textrm{con}_{jk} - 1\right|\Delta t = \sum_{k=1}^{t_f/\Delta t} \left|\sum_{j=1}^N u^\textrm{con}_{jk} - 1\right|\Delta t.
\end{align}
In the following theorem, we present the estimation of the difference in objective values between continuous and binary controls. 
\begin{theorem}
\label{theo:str-no-sos1}
Let $u^\textrm{con}$ be the solution of the continuous relaxation of model~\eqref{eq:model-d-1} without the SOS1 property under fixed time discretization. 
Let $u^\textrm{bin}$ be the binary control with the enforced SOS1 property obtained by Algorithm~\ref{alg:switch-eobj} with $\alpha=0$. 
There exist constants $C_0,\ C_1,\ C_2$ such that the difference between objective values of continuous and binary controls satisfies
\begin{align}
    \label{eq:diff-no-sos1}
    \bar{F}(u^\textrm{bin}) - \bar{F}(u^\textrm{con})\leq 2C_1e^{C_2 \Delta t}\Delta t + C_0\epsilon^c(\Delta t).
\end{align}
\end{theorem}
We provide detailed proof in Appendix~\ref{app:proof}. 
We consider a uniform discretized model with $L_2$ penalty function for the SOS1 property with the following formulation~\citep{Fei2023binarycontrolpulse}:
\begin{align}
\label{eq:model-l2} \tag{P-$L_2$}
    \min_{u\in [0,1]^{N\times T},X,H} \quad & \bar{F}(u) + \rho \sum_{k=1}^T \left(\sum_{j=1}^N u_{jk} - 1\right)^2,
\end{align}
where $\rho$ is the penalty parameter. We assume that the original problem~\eqref{eq:model-c-1} has a feasible solution with the SOS1 property. 
In the following corollary, we show that if we solve the model with $L_2$ penalty function~\eqref{eq:model-l2}, then the convergence result in Theorem~\ref{theo:str} still holds. 
\begin{corollary}
\label{cor:str}
Assume that the original problem~\eqref{eq:model-c-1} has a feasible solution with the SOS1 property. 
For any $\Delta t$, let the discretized continuous controls $u^\textrm{con}$ be the solution of model~\eqref{eq:model-l2} and $u^\textrm{bin}$ be the binary solutions obtained by Algorithm~\ref{alg:switch-eobj} with $\alpha=0$, then we have
\begin{align}
    \limsup\limits_{\Delta t \rightarrow 0}\bar{F}(u^\textrm{bin}) - \bar{F}(u^\textrm{con})\leq 0. 
\end{align}
\end{corollary}
\begin{proof}
    It is obvious that all the constants $C_0,\ C_1,\ C_2$ are bounded. We only need to prove that $\displaystyle \limsup_{\Delta t\rightarrow 0} \epsilon^c(\Delta t) = 0$. 
    With the assumption that the original problem~\eqref{eq:model-c-1} has a feasible solution with the SOS1 property, we have that there exist continuous functions $u_j(t),\ j=1,\ldots,N$ such that $\displaystyle \sum_{j=1}^N u_{j}(t)=1$ for $t\in [0,t_f]$ almost everywhere. For any time discretization with $T$ time intervals and length $\Delta t=t_f/T$, we construct a discretized solution $u^{(T)}$ such that 
    \begin{align}
        u_{jk}^{(T)} = \int_{t_{k-1}}^{t_k} u_j(t) dt/\Delta t,\ k=1,\ldots,T.
    \end{align}
    With the definition, we have $\displaystyle \sum_{j=1}^N u_{jk}^{(T)}=1,\ k=1,\ldots,T$. 
    Let $u^{(T)*}$ be the optimal solution of problem~\eqref{eq:model-l2} with time steps $T$, then we have
    \begin{align}
        \bar{F}(u^{(T)*})+\rho T \sum_{k=1}^T \left(\sum_{j=1}^N u_{jk}^{(T)*} - 1\right)^2 \frac{t_f}{T}\leq \bar{F}(u^{(T)})
    \end{align}
    from the definition of the optimal solution. Taking the limit when $T\rightarrow \infty$, i.e. $\Delta t\rightarrow 0$, we have 
    \begin{align}
    \label{eq:f-ub}
        \lim_{T\rightarrow \infty} \bar{F}(u^{(T)*})+\lim_{T\rightarrow \infty}\rho T \sum_{k=1}^T \left(\sum_{j=1}^N u_{jk}^{(T)*} - 1\right)^2 \frac{t_f}{T}\leq \lim_{T\rightarrow \infty}\bar{F}(u^{(T)})\leq F_{UB}, 
    \end{align}
    where $F_{UB}$ is the upper bound of the objective function $F$ and the last inequality follows $\displaystyle \lim_{T\rightarrow \infty} u_j^{(T)}=u_j(t)$.
    
    With the definition that $u^{\textrm{con}}=u^{(T)*}$, we have 
    \begin{align}
    \label{eq:norm-inequality}
        \sum_{k=1}^T \left|\sum_{j=1}^N u_{jk}^{\textrm{con}} - 1\right| \leq   \sqrt{T} \left(\sum_{k=1}^T \left(\sum_{j=1}^N u_{jk}^{(T)*} - 1\right)^2\right)^{1/2}
    \end{align}
    for any $T$ following the norm inequality between $L_1$ and $L_2$ norm. 
    Consider two cases that the $L_2$ norm in equation~\eqref{eq:norm-inequality} is smaller or not smaller than $1$, then we have 
    \begin{align}
    \label{eq:epsilon-inequality}
        \sum_{k=1}^T \left|\sum_{j=1}^N u_{jk}^{\textrm{con}} - 1\right| \leq    \max\{\sqrt{T},\ \sqrt{T} \sum_{k=1}^T \left(\sum_{j=1}^N u_{jk}^{(T)*} - 1\right)^2\}
    \end{align}
    for any $T$, respectively. 
    Next we prove $\displaystyle \limsup_{\Delta t\rightarrow 0} \epsilon^c(\Delta t)=0$. Because $\epsilon^c(\Delta t)$ can be written as
    \begin{align}
        \epsilon^c(\Delta t) = \sum_{k=1}^T \left|\sum_{j=1}^N u_{jk}^{\textrm{con}} - 1\right|\frac{t_f}{T}, 
    \end{align}
     we only need to prove $\displaystyle \lim_{T\rightarrow \infty} \sqrt{T}\frac{t_f}{T}=0$ and $\displaystyle \lim_{T\rightarrow \infty} \sqrt{T} \sum_{k=1}^T\left(\sum_{j=1}^N u_{jk}^{(T)*} - 1\right)^2 \frac{t_f}{T}=0$.

    With $t_f$ as a constant, we directly have $\displaystyle \lim_{T\rightarrow \infty} \sqrt{T}\frac{t_f}{T}=0$. We prove the second limit by contradiction. 
    Assume there exists a constant $\epsilon_0>0$ such that $\displaystyle \sqrt{T}\sum_{k=1}^T \left(\sum_{j=1}^N u_{jk}^{(T)*} - 1\right)^2 \frac{t_f}{T}\geq \epsilon_0>0$ for any $T$, then the quadratic penalty term in the objective function of~\eqref{eq:model-l2}
    \begin{align}
        \lim_{T\rightarrow \infty}\rho T \sum_{k=1}^T \left(\sum_{j=1}^N u_{jk}^{(T)*} - 1\right)^2 \frac{t_f}{T} = \lim_{T\rightarrow \infty}\rho \sqrt{T} \sqrt{T} \sum_{k=1}^T \left(\sum_{j=1}^N u_{jk}^{(T)*} - 1\right)^2 \frac{t_f}{T} \geq  \lim_{T\rightarrow \infty}\rho \sqrt{T} \epsilon_0 = \infty, 
    \end{align}
    which contradicts the upper bound $F_{UB}$ in~\eqref{eq:f-ub}. Therefore, we have
    \begin{align}
        \limsup_{\Delta t\rightarrow 0} \epsilon^c(\Delta t)\leq \limsup_{T\rightarrow \infty} \max\{\sqrt{T} \frac{t_f}{T},\ \sqrt{T} \sum_{k=1}^T \left(\sum_{j=1}^N u_{jk}^{(T)*} - 1\right)^2 \frac{t_f}{T}\}=0.
    \end{align}
\end{proof}
In the following corollary, we prove that for quantum control problems without the SOS1 property, the convergence results still hold. 
\begin{corollary}
\label{cor:str-nosos1}
Let $u^\textrm{con}$ be the solution of the continuous relaxation of model~\eqref{eq:model-d-1} without the SOS1 property. Let $u^\textrm{bin}$ be the binary control obtained by Algorithm~\ref{alg:switch-eobj} without the SOS1 property constraint (i.e. $\mathcal{U}=\{0,1\}^{N}$) setting $\alpha=0$. Then we have
\begin{align}
    \limsup\limits_{\Delta t \rightarrow 0}\bar{F}(u^\textrm{bin}) - \bar{F}(u^\textrm{con})\leq 0. 
\end{align}
\end{corollary}
\begin{proof}
    We complete the proof by showing that we can convert a continuous solution $u$ without the SOS1 property to a solution $u'$ with the SOS1 property with re-defined Hamiltonian controllers. 
    
    Specifically, for a control system with $N$ Hamiltonian controllers $H^{(1)},\ldots,H^{(N)}$, 
    we introduce a new control system with $2^N$ controllers, 
    including Hamiltonian controllers $\displaystyle \hat{H}^{(l_1,\ldots, l_i)} = \sum_{j=1}^{i} H^{(l_j)},\ \forall \left\{l_1,\ldots,l_i\right\}\subseteq \left\{1,\ldots,N\right\},\ \forall i=1,\ldots,N$, Especially $\hat{H}^{(\emptyset)}$ is an idle controller with represented as an all-zero Hamiltonian matrix. 
For any given controls $u_k\in [0,1]^N$ of the continuous relaxation of the discretized model~\eqref{eq:model-d-1} without the SOS1 property, we can convert them to controls $u'_k\in [0,1]^{2^N}$ such that $\sum_{j=1}^{2^N} u'_{jk}=1$ as the following rule for every time step $k=1,\ldots,T$. 
At time step $k$, given control variable $u_k$ for $N$ controllers, we sort the control values by descending order as $u_{l_1(k)k}\geq \ldots\geq u_{l_N(k)k}$ where $l_i(k)$ is the controller index with the $i$th high control value. For simplicity, we represent the sorted controller index as $l_1,\ldots,l_N$ by eliminating $k$, then we reconstruct controls $u'$ as 
\begin{subequations}
\begin{align}
    & u'_{(l_1,\ldots, l_i)k} = u_{l_ik} - u_{l_{i+1}k},\ i=1,\ldots,N-1\\
    & u'_{(l_1, \ldots, l_N)k} = u_{l_Nk}\\
    & u'_0 = 1 - u_{l_1k}, 
\end{align}
\end{subequations}
where $u'_{(l_1,\ldots, l_i)}$ is the control value corresponding to Hamiltonian controller $H^{(l_1,\ldots, l_i)},\ i=1,\ldots,N$ and $u'_0$ is the control value corresponding to all-zeros Hamiltonian controller $\hat{H}^{(\emptyset)}$. We have that
\begin{align}
    \sum_{i=1}^N u'_{(l_1\ldots l_i)k} \hat{H}^{(l_1\ldots l_i)} + u'_0 \hat{H}^{(\emptyset)} 
    & = \sum_{i=1}^N u'_{(l_1\ldots l_i)k} \sum_{j=1}^i H^{(l_j)} 
    =\sum_{j=1}^N \left(\sum_{i=j}^N u'_{(l_1\ldots l_i)k}\right) H^{(l_j)} \nonumber \\
    & =\sum_{j=1}^N u_{l_jk}H^{(l_j)} = \sum_{l=1}^N u_{lk} H^{(l)}. 
\end{align}
Therefore the two control variables $u$ and $u'$ with corresponding control systems have the same impact on the states. 
Based on the definition of $u'$ and bound constraints that $u\in [0,1]^{N\times T}$, we have $\displaystyle \sum_{i=1}^N u'_{(l_1\ldots l_i)k} + u'_0=1$ and $u'\in [0,1]^{2^N\times T}$, which means that the SOS1 property holds for the new control system. 
From Corollary~\ref{cor:str}, the convergence results of Algorithm~\ref{alg:switch-eobj} hold for $u'$ with the SOS1 property, hence, hold for $u$.
\end{proof}

\subsection{Method Based on Cumulative Difference}
\label{sec:alg-sur}
Evaluating the difference between binary controls and continuous controls by cumulative difference was proposed by~\citet{sager2011combinatorial} and is widely used in optimal control theory~\citep{sager2012integer, you2011mixed, manns2020multidimensional}. 
In this section, we propose another method (Algorithm~\ref{alg:switch-cumdiff}) to balance the cumulative difference between binary controls and continuous controls and the penalized number of switches. We define $\beta$ as the TV norm penalty parameter. The explanation of the algorithm is as follows. 
At each time step $k$, we compute the cumulative difference of current controls $[u_1^\textrm{bin},\ldots,u_{k-1}^\textrm{bin}, \hat{u}]\in [0,1]^{N\times T}$ for all $\hat{u}\in \mathcal{U}$ and select $\hat{u}^*$ as the one with the smallest cumulative difference (see step~\ref{algline:sur-diff}--\ref{algline:sur-select}). 
If the cumulative difference of keeping control of the current time step the same as the previous time step, i.e. $\hat{u}=u_{k-1}$, is no larger than the weighted TV-norm value of choosing the control with the smallest cumulative difference, i.e. $\hat{u}=\hat{u}^*$, we keep the previous time step controllers (see step~\ref{algline:sur-condition}--\ref{algline:sur-keep}). Otherwise, we choose the controller with the smallest cumulative difference. Then update the binary control $u^\textrm{bin}$ and move to the next time step $k+1$ (see step~\ref{algline:sur-update}).  

According to the paper by~\citet{sager2011combinatorial}, the cumulative difference at each time step $k$ for any feasible control $\hat{u}\in \mathcal{U}$ is defined by
\begin{align}
\label{eq:cum-diff}
    \textrm{Diff}\left(u^\textrm{con}, [u_1^\textrm{bin},\ldots,u_{k-1}^\textrm{bin}, \hat{u}]\right) = \max_{j=1,\ldots,N} \ \left|\sum_{l=1}^{k-1} \left(u_{jl }^\textrm{con}  - u_{jl}^{bin}\right) \Delta t + \left(u_{jk }^\textrm{con}  - \hat{u}_j\right)\Delta t \right|.
\end{align}
Following the paper by~\citet{sager2012integer}, we introduce auxiliary deviation variables $\hat{p}_{jk}$ for each controller $j=1,\ldots,N$ and each time step $k=1,\ldots,T$, such that:
\begin{align}
    \label{eq:deviation}
        \hat{p}_{jk} = \sum_{l=1}^{k} u^\textrm{con}_{jl} \Delta t - \sum_{l=1}^{k-1} u^\textrm{bin}_{jl} \Delta t,\ j=1,\ldots,N,\ k=1,\ldots,T.
\end{align}
At time step $k=1,\ldots,T$, we use $\hat{p}_k$ to represent the column vector of deviation variables. We update deviation variables by a recursive formula as:
\begin{align}
    \label{eq:deviation-update}
    \hat{p}_k = \hat{p}_{k-1} + u_k^\textrm{con} - u_{k-1}^\textrm{bin},\ k=2,\ldots,T.
\end{align} 
For any binary control $\hat{u}$, the cumulative difference at time step $k$ can be computed by the deviation variables as
\begin{align}
    \textrm{Diff}\left(u^\textrm{con}, [u_1^\textrm{bin},\ldots,u_{k-1}^\textrm{bin}, \hat{u}]\right) = 
    \max\{\max_{j:\hat{u}_j=1} \left|\hat{p}_{jk} - \Delta t \right|, \max_{j:\hat{u}_j=0} \left|\hat{p}_{jk}\right|\}.
\end{align}
Therefore we only require to compute $\hat{p}_{jk}$ and $\hat{p}_{jk}-\Delta t$ for each time step. 
With the definition of deviation variables, we conclude that setting controller $j$ active does not increase the cumulative difference if and only if $|\hat{p}_{jk}| > |\hat{p}_{jk}-\Delta t|$, which means that $\hat{p}_{jk}\geq 0.5\Delta t$. Furthermore, for any two controllers, setting the one with a higher deviation leads to a higher decrease or lower increase in the cumulative difference. 
\begin{remark}
\label{rmk:cum-diff}
   If binary control $u^\textrm{bin}$ requires the SOS1 property, at each time step $k$, let $j^*=\argmax_{j=1,\ldots,N} \hat{p}_{jk}$, 
    we choose $\hat{u}^*_k$ as $\hat{u}^*_{j^*k} = 1,\ \hat{u}^*_{jk}=0,\ \forall j\neq j^*$. If binary control does not require any additional constraints, at each time step $k$, 
    let $J^*=\{j=1,\ldots,N: \hat{p}_{jk}\geq 0.5\Delta t\}$, 
    we choose $\hat{u}^*_k$ as $\hat{u}^*_{jk} = 1,\ \forall j\in J^*,\ \hat{u}^*_{jk}=0,\ \forall j\notin J^*$. 
\end{remark}

We describe the detailed algorithm in Algorithm~\ref{alg:switch-cumdiff}. 
\begin{algorithm}[!ht] \caption{Heuristic Rounding Method based on Cumulative Difference (Cdiff) \label{alg:switch-cumdiff}} 
    \DontPrintSemicolon
    \SetNoFillComment
    \KwInput{Continuous control $u^\textrm{con}$ and TV penalty parameter $\beta$.}
    Initialize binary control 
    $\hat{p}_1=u^\textrm{con}_{1}$.\;
    Let $u_1^\textrm{bin} = \argmin_{\hat{u}\in \mathcal{U}} \textrm{Diff}\left(u^\textrm{con}, [\hat{u}]\right)$.\;
    \For{$k=2,\ldots,T$}
    {
    $\hat{p}_{k}=\hat{p}_{k-1} + u^\textrm{con}_{k} - u^\textrm{bin}_{k-1}$.\;
    \label{algline:sur-diff}
    Let $\hat{u}^*=\argmin_{\hat{u}\in \mathcal{U}} \textrm{Diff}\left(u^\textrm{con}, [u_1^\textrm{bin},\ldots,u_{k-1}^\textrm{bin}, \hat{u}]\right)$.\;
    \label{algline:sur-select}
    \eIf{
    $\textrm{Diff}\left(u^\textrm{con}, [u_1^\textrm{bin},\ldots,u_{k-1}^\textrm{bin}, u_{k-1}^\textrm{bin}]\right) \leq \beta TV\left([u_1^\textrm{bin},\ldots,u_{k-1}^\textrm{bin}, \hat{u}^*, u_{k+1}^\textrm{con},\ldots,u_{T}^\textrm{con}]\right)$
    \label{algline:sur-condition}
    }
    {
    Update binary control $u_k^\textrm{bin} = u_{k-1}^\textrm{bin}$.
    \label{algline:sur-keep} }
    {
    Update control $u_k^\textrm{bin} = \hat{u}^*$, breaking ties
    choosing control with smallest TV value.\;
    \label{algline:sur-update}}
    }
    \KwOutput{Binary control $u^\textrm{bin}$}
\end{algorithm}
Notice that the maximum number of switches with penalty parameter $\beta$ is $t_f/\beta$ following from the definition of cumulative differences. When the weight parameter $\beta=0$, the convergence results of SUR~\citep{sager2012integer,Fei2023binarycontrolpulse} ensure that the cumulative difference between $u^\textrm{bin}$ and $u^\textrm{con}$ converges to zero, no matter whether the SOS1 property holds or not,  leading to the objective value of the binary control $\bar{F}(u^\textrm{bin})$ converges to the objective value of the continuous control $\bar{F}(u^\textrm{con})$ when $\Delta t$ goes to zero. 

%% file: chpt3_switching_optimization.tex
\section{Switching Time Optimization}
\label{sec:model-switching}
With derived controller sequences in Section~\ref{sec:alg-extraction}, we next propose our generic switching time optimization model based on the continuous time horizon. Then we introduce the solution method and an acceleration technique for time-evolution simulations. 
\paragraph{Formulation} Let $S$ be the total number of controllers in the controller sequence $\bar{\mathcal{H}}$, we divide the whole time horizon into $S$ time intervals as $0=t_0<t_1<\ldots<t_S=t_f$ where each time interval has a fixed Hamiltonian controller. 
For each time interval $s=1,\ldots,S$, we define variables $\tau_s\in [0,t_f]$ as the length and $X_s\in \mathbb{C}^{2^q\times 2^q}$ as the unitary operator at the end of the time interval. 
Let $\tau$ be the corresponding vector form of $\tau_s$. By definition, the final unitary operator $X(t_f) = X_S$. 
The start and end time of time interval $s$ is computed by $t_{s-1}=\sum_{l=1}^{s-1} \tau_l$ and $t_s = \sum_{l=1}^s \tau_l$. 
The unitary operators follow the differential equation 
\begin{align}
    \frac{d}{dt} X(t) = -i\bar{H}_s X(t),\ \forall t\in (t_{s-1}, t_s],\ s=1,\ldots,S.
\end{align}
We obtain the explicit solution as $X(t)=\exp\{-i\bar{H}_s (t - t_{s-1})\}X(t_{s-1}),\ \forall t\in (t_{s-1}, t_s],\ s=1,\ldots,S$. 
As a result, the final operator $X(t_f)$ can be computed as an implicit function of $\tau$:
\begin{align}
    \label{eq:final-state-s}
    X(t_f, \tau) = X_S(\tau) = \prod_{s=1}^S e^{-i\bar{H}_s \tau_s} X_\textrm{init}.
\end{align}
Substituting the final operator $X(t_f, \tau)$ into a general objective function $F$, we obtain an objective function $\hat{F}(\tau)$ with respect to variables $\tau$ and formulate the switching time optimization problem as
\begin{subequations}
\makeatletter
\def\@currentlabel{STO}
\makeatother
\label{eq:model-ts}
\begin{align}
    \label{eq:model-ts-obj}
    \textrm(STO)\quad \min_{\tau} \quad & \hat{F}(\tau) = F(X(t_f, \tau))\\
    \label{eq:model-ts-sum-tau}
    \textrm{s.t.} \quad & \sum_{s=1}^{S} \tau_s = t_f \\
    \label{eq:model-ts-tau}
    & \tau_s\geq 0,\ s=1,\ldots,S.
\end{align}
\end{subequations}
% The objective function~\eqref{eq:model-ts-obj} is derived by substituting the final operator $X(t_f)$~\eqref{eq:final-state-s} into original objective functions presented in~\eqref{eq:obj-energy}--\eqref{eq:obj-fid}, which is only respect to the interval length variables $\tau$. 
The objective function~\eqref{eq:model-ts-obj} is a general function only corresponding to $\tau$.
Constraint~\eqref{eq:model-ts-sum-tau} enforces that the sum of the time intervals equals the evolution time. Constraint~\eqref{eq:model-ts-tau} ensures that the length of each time interval is in $[0,t_f]$. Note that time interval $s$ with length as zero means that we do not apply Hamiltonian controller $\bar{H}_s$ to the control system, leading to a reduction in the number of switches. 

\paragraph{Solution method} Because the switching optimization model~\eqref{eq:model-ts} includes the additional summation constraint~\eqref{eq:model-ts-sum-tau}, we derive the gradient of $\hat{F}(\tau)$ as follows and solve the model by the sequential least-squares programming (SLSQP) algorithm~\citep{2020SciPy-NMeth}, which is a widely used iterative algorithm for solving bounded optimization problems with equality constraints. 
For each time interval $s=1,\ldots,S$, we define propagators $U_s
= \exp\{-i\bar{H}_s\tau_s\}$. The gradient of $U_s$ corresponding to $\tau_s$ is computed as  
\begin{align}
    \frac{\partial U_s}{\partial \tau_s}= -i \bar{H}_s U_s,\ s=1,\ldots,S.
\end{align} 
The gradient of the objective function with respect to the propagators varies with its specific formulation.   
For the energy function, we define back propagation variables $\kappa_s = U_{s+1}^\dagger \cdots U_{S}^\dagger \tilde{H}X_S |\psi_0\rangle$ for each interval $s=1,\ldots,S-1$ and $\kappa_S=\tilde{H}X_S |\psi_0\rangle$. 
The gradient with respect to $\tau_s$ is computed as 
\begin{align}
    \frac{\partial \hat{F}}{\partial \tau_s}= \frac{2}{E_\textrm{min}}\textrm{Re}\left[i \langle \kappa_{s}| \bar{H}_s X_s |\psi_0\rangle\right], \ s=1,\ldots,S.
\end{align}
For the infidelity function, 
we define back propagation variables $\lambda_s = U_{s+1}^\dagger \cdots U_S^\dagger X_\textrm{targ}$ for each interval $s=1,\ldots,S-1$ and $\lambda_S = X_\textrm{targ}$. Then the gradient of the trace with respect to $\tau_s$ is computed as 
\begin{align}
    \frac{\partial \operatorname{tr}\left\{X_{\textrm{targ}}^{\dagger} X_S\right\}}{\partial \tau_s}= -i\lambda_{s}^\dagger \bar{H}_s X_s,\ s=1,\ldots,S.
\end{align}
Using the definition of the infidelity objective function~\eqref{eq:obj-fid}, we compute the gradient as  
\begin{align}
    \frac{\partial \hat{F}}{\partial \tau_s}= \frac{1}{2^q}\textrm{Re}\left[i\mathbf{tr}\left\{\lambda_{s}^\dagger \bar{H}_s X_s\right\} e ^{-i\textrm{arg} \left(\operatorname{tr}\left\{X_{\textrm{targ}}^{\dagger} X_S\right\}\right)}\right],\ s=1,\ldots,S, 
\end{align}
where $\textrm{arg}(\cdot)$ represents the argument of a complex number. 

\paragraph{Time-evolution simulation} For each set of switching time points, we can input controllers into quantum systems and conduct the time-evolution process to obtain final objectives and gradients by the finite difference method. 
Conducting time-evolution simulations on classical computers is still a widely used approach in quantum control research. Directly computing matrix exponentials during time evolution requires high computational costs. 
Therefore, we introduce an acceleration technique of time-evolution simulations with pre-computed eigenvalue decompositions. Because all the Hamiltonian controllers in quantum systems are Hermitian matrices, they can be written as
\begin{align}
    \bar{H}_s = Q_s\Lambda_s Q_s^\dagger,\ s=1,\ldots, S
\end{align}
where $\Lambda_s=\textbf{diag}\{\lambda_{1s},\ \ldots,\ \lambda_{2^qs}\}$ are diagonal matrices with eigenvalues $\lambda_{1s},\ \ldots,\ \lambda_{2^qs}$ of $\bar{H}_s$ as diagonals and $Q_s$ are unitary matrices. According to the paper by~\citet{hall2013lie}, the corresponding propagators are computed as
\begin{align}
    \label{eq:reform-matrix-exponential}
        U_s = e^{-i\bar{H}_s\tau_s} =  Q_se^{-i\Lambda_s \tau_s}Q_s^\dagger = Q_s\textbf{diag}\{e^{\lambda_{1s}},\ \ldots,\ e^{\lambda_{2^qs}}\}Q_s^\dagger,\ s=1,\ldots,S.
\end{align}
Therefore, with pre-decomposed $\Bar{H}_s,\ s=1,\ldots,S$, for any interval length variables $\tau$, we only need to compute $2^q$ real number exponentials for all the eigenvalues and $2$ matrix exponentials for computing each propagator during the time-evolution process. The number of eigenvalue decompositions we require to pre-compute is no more than the number of distinctive controllers in the sequence $\bar{\mathcal{H}}$, which is at most $N$ for problems with the SOS1 property and $2^N$ for problems without the SOS1 property. 

%% file: chpt4_numeric.tex
\section{Numerical Studies}
\label{sec:numerical}
We apply our solution frameworks from Sections~\ref{sec:alg-extraction} and \ref{sec:model-switching} with switching time optimization to solve multiple specific quantum control instances, to demonstrate their computational efficacy. 
In Section~\ref{sec:instance}, we introduce our experimental design and four specific problems. 
In Section~\ref{sec:sensitivity}, we select a specific instance to conduct the sensitivity analysis of the switching penalty parameter for obtaining binary controls. 
In Section~\ref{sec:results}, we first present the objective value results of different methods and then show the computational time results and the acceleration of our pre-computing matrix decomposition technique. 
We demonstrate that our switching time optimization framework eliminates the requirement for precise time discretization. In Section~\ref{sec:control}, we discuss the optimal control figures of selected instances, to provide intuitive insights for our algorithms. 

\subsection{Simulation Design}
\label{sec:instance}
For all the instances, we first solve the continuous relaxation of the discretized model~\eqref{eq:model-d-1} to obtain continuous solutions. Then we apply three methods to obtain discretized binary controls and their corresponding controller sequences. 
The first one~\citep[Algorithm TR+MT+ALB]{Fei2023binarycontrolpulse} solves the continuous relaxation with TV regularizer by a trust-region method (TR), then rounds continuous solutions with min-up-time constraints (MT), and improves binary solutions by an approximate local-branching method (ALB). 
We select TR+MT+ALB as the benchmark algorithm because it obtains the best trade-off between objective values and TV regularizer values among all the instances. 
The second and third methods are heuristic methods based on objective value (Algorithm~\ref{alg:switch-eobj}) and cumulative difference (Algorithm~\ref{alg:switch-cumdiff}). 
With obtained controller sequences, we solve the switching time optimization model~\eqref{eq:model-ts} and obtain final solutions. 
We introduce four specific quantum control problems as follows. 

\paragraph{Energy minimization problem}
We consider a quantum spin system with $q$ qubits, no intrinsic Hamiltonian $H^{(0)}$, and two control Hamiltonians $H^{(1)},\ H^{(2)}$. The initial state $|\psi_0\rangle$ is the ground state of $H^{(1)}$, defined as the eigenvector with minimum eigenvalue. 
The initial operator $X_{\textrm{init}}$ is a $2^q$-dimensional identity matrix. Our goal is to minimize the system energy corresponding to $H^{(2)}$, with minimum energy $E_{\textrm{min}}<0$ as the ground state energy of $H^{(2)}$. 
The objective function is the energy ratio function~\eqref{eq:obj-energy} with $\Tilde{H}=H^{(2)}$. 
We use $\sigma_i^x,\ \sigma_i^z$ to represent the Pauli matrices of qubit $i$ for $i=1,\ldots,q$. 
We define a matrix $[J_{ij}],\ i,j=1,\ldots,q$ as the coupling matrix for the instances of all-to-all spin glasses.  
\begin{subequations}
\label{eq:model-energy-c}
\begin{align}
\label{eq:model-energy-c-obj}
    \min_{u(t), X(t), H(t)}\quad  \displaystyle & 1 - \left \langle \psi_0\right| X(t_f)^\dagger H^{(2)} X(t_f) \left|\psi_0 \right\rangle / E_{\textrm{min}}\\
    \textrm{s.t.}\quad %
                    & H(t) = u_1(t)H^{(1)} + u_2(t)H^{(2)},\ \forall t\in [0,t_f]\\
                    \label{eq:model-energy-h1}
                    & H^{(1)} = -\sum_{i=1}^q \sigma_i^x
                    ,\ H^{(2)} = \sum_{ij} J_{ij}\sigma_i^z \sigma_j^z\\
                     & \frac{d}{dt}X(t)=-iH(t)X(t)\nonumber\\
                    & X(0)= X_\textrm{init} \nonumber\\
                    \label{eq:model-energy-cons-u-sum}
                    & u_1(t)+u_2(t)=1\\
                    & u_1(t),\ u_2(t) \in \left\{0,1\right\}, 
\end{align}
\end{subequations}
where constraint~\eqref{eq:model-energy-cons-u-sum} indicate that the SOS1 property is required in this problem.
We test instances for quantum systems with qubit numbers $q=2,\ 4,\ 6$ and represent the instances by Energy2, Energy4, and Energy6, respectively. When $q=2$, $J$ is fixed as a matrix with zero diagonals and other elements as one. When $q=4,\ 6$, $J$ is a symmetric matrix with zero diagonals and other elements randomly generated from the interval $[-1,1]$. 

\paragraph{CNOT estimation problem}
This problem is defined on a quantum system modeled as an isotropic Heisenberg spin-1/2 chain with length 2~\citep{pawela2016various}. We only consider the unitary evolution without energy leakage. 
The system includes an intrinsic Hamiltonian $H^{(0)}$ given by the Heisenberg Hamiltonian, and two Zeeman-like control Hamiltonians $H^{(1)},\ H^{(2)}$ which are only performed on the first spin. We define the initial operator $X_\textrm{init}$ as a 4-dimensional identity matrix and the target operator as the matrix formulation of the CNOT gate as: 
\begin{equation*}
    X_{\textrm{CNOT}} = \begin{pmatrix}
        1 & 0 & 0 & 0\\
        0 & 1 & 0 & 0\\
        0 & 0 & 0 & 1\\
        0 & 0 & 1 & 0
    \end{pmatrix}.
\end{equation*}
With $\sigma_i^x,\ \sigma_i^y,\ \sigma_i^z$ as the Pauli matrices for qubit $i=1,2$, we have the following formulation:
\begin{subequations}
\label{eq:model-cnot-c}
\begin{align}
    \min_{u(t), X(t), H(t)}\quad  & 1 - \frac{1}{4}\left|\operatorname{tr}\left\{X_\textrm{CNOT}^\dagger X(t_f)\right \}\right| \\
    \textrm{s.t.}\quad  & H(t) = H^{(0)} + u_1(t)H^{(1)} + u_2(t)H^{(2)},\ \forall t\in [0,t_f]\\
                    \label{eq:model-cnot-h0}
                    & H^{(0)} = \sigma^x_{1} \sigma^x_{2}+\sigma^y_{1} \sigma^y_{2}+\sigma^z_{1} \sigma^z_{2},\ H^{(1)} = \sigma^x_1
                    ,\ H^{(2)} = \sigma^y_1\\
                    & \frac{d}{dt}X(t)=-iH(t)X(t)\nonumber\\
                    & X(0)= X_\textrm{init} \nonumber\\
                    & u_1(t),\ u_2(t) \in \left\{0,1\right\}.
\end{align}
\end{subequations}
The objective function is the infidelity function~\eqref{eq:obj-fid} with $X_\textrm{targ}=X_\textrm{CNOT}$. To demonstrate the generality of our algorithm, we relax the SOS1 property in this problem. We test instances with evolution times $t_f=5,\ 10,\ 20$ and label them by CNOT5, CNOT10, and CNOT20, respectively.

\paragraph{NOT estimation problem}
Following the setting in the paper by~\citet{Motzoi2009}, we consider the lowest three levels of a driven slightly anharmonic energy spectrum for a single qubit, where the nearest levels have couplings. The first two levels comprise the qubit. The third level is a physical state representing the energy leakage but not a logical state. 
The system contains an intrinsic Hamiltonian $H^{(0)}$ and two control Hamiltonians $H^{(1)},\ H^{(2)}$. We define the initial operator $X_\textrm{init}$ as a 3-dimensional identity matrix and the target operator as the matrix formulation of the NOT gate:
\begin{align*}
    X_{\textrm{NOT}} = \begin{pmatrix}
    0 & 1 & 0\\
    1 & 0 & 0\\
    0 & 0 & 0\\
    \end{pmatrix}.
\end{align*}
For level $j=0,\ 1,\ 2$, the quantum state $|j\rangle$ is represented by a standard vector $j$ with the $j$th element as $1$ and other elements as $0$. 
The continuous-time formulation of this problem is 
\begin{subequations}
\label{eq:model-not-c}
\begin{align}
    \min_{u(t), X(t), H(t)}\quad  & 1 - \frac{1}{2}\left|\operatorname{tr}\left\{X_\textrm{NOT}^\dagger X(t_f)\right \}\right| \\
    \textrm{s.t.}\quad  & H(t) = H^{(0)} + u_1(t)H^{(1)} + u_2(t)H^{(2)},\ \forall t\in [0,t_f]\\                
    \label{eq:model-not-h0}
                    & H^{(0)} = J_1 |1\rangle \langle 1| + J_2 |2\rangle \langle 2|\\ 
                    \label{eq:model-not-h1}
                    & H^{(1)} = {\omega_1}(|0\rangle \langle 1| + |1\rangle \langle 0|)
                    + {\omega_2}(|1\rangle \langle 2| + |2\rangle \langle 1|)\\
                    \label{eq:model-not-h2}
                    & H^{(2)} = {\omega_1}(i|0\rangle \langle 1| -i |1\rangle \langle 0|)
                    + {\omega_2}(i|1\rangle \langle 2| -i |2\rangle \langle 1|)\\
                    & \frac{d}{dt}X(t)=-iH(t)X(t)\nonumber\\
                    & X(0)= X_\textrm{init} \nonumber\\
                    & u_1(t),\ u_2(t) \in \left\{0,1\right\}, 
\end{align}
\end{subequations}
where $J_1,\ J_2$ are the relative strength of transitions between different qubits. Without loss of generality, we choose  $J_1=1,\ J_2=\sqrt{2}$.
The parameters $\omega_1,\ \omega_2$ are the detunings of
the transitions with respect to the drive frequency. We assume that the drive frequency is resonant with the qubit frequency, which means that $\omega_1=0,\ \omega_2=2\pi$. 
Keeping consistent with common quantum system settings with energy leakage~\citep{Motzoi2009,rebentrost2009optimal}, we also relax the SOS1 property in this problem.
We test instances with evolution times $t_f=2,\ 6,\ 10$ and represent them by NOT2, NOT6, and NOT10, respectively.

\paragraph{Circuit compilation problem}
The goal of the circuit compilation problem is to represent a given quantum circuit by specific controllers dependent on the quantum system. 
We consider a gmon qubit quantum system, which is a superconducting qubit architecture that combines high-coherence qubits and tunable qubit couplings~\citep{chen2014qubit}. 
In this system, qubits are connected to their nearest neighbors with a rectangular-grid topology. We use $E$ to represent the connected qubits.
Each qubit has two Hamiltonian controllers, a flux-drive controller, and a charge-drive controller, whose control functions are represented by $u(t)$.  
Every two connected qubits have a corresponding Hamiltonian controller, whose control functions are represented by $w(t)$. 
Following the paper by~\citet{gokhale2019partial}, we take quantum circuits generated by the Unitary Coupled-Cluster Single-Double (UCCSD) methodology~\citep{bartlett2007coupled, romero2018strategies} for estimating the ground state energy of molecules in quantum chemistry as examples. 
For a system with $q$ qubits, the initial operator $X_{\textrm{init}}$ is a $2^q$-dimensional identity matrix and the target operator is the matrix formulation of the quantum circuit, represented by $X_{\textrm{UCCSD}}\in \mathbb{C}^{2^q\times 2^q}$. The specific problem formulation is given as
\begin{subequations}
\label{eq:model-molecule-c}
\begin{align}
    \min_{u(t), X(t), H(t)}\quad  & 1 - \frac{1}{2^q}\left|\operatorname{tr}\left\{X_\textrm{UCCSD}^\dagger X(t_f)\right \}\right| \\
    \textrm{s.t.}\quad  & H(t) = \sum_{j=1}^{2q} u_j(t)H^{(j)} + \sum_{(j_1,j_2)\in E} w_{j_1j_2}(t)H^{(j_1j_2)},\ \forall t\in [0,t_f]\\
                    \label{eq:model-molecule-h1}
                    & H^{(2j-1)} = J_c \sigma^x_{j},\ 
                    H^{(2j)} = J_f\left(\begin{array}{ll}
                                0 & 0 \\
                                0 & 1
                    \end{array}\right)_{(j)},\ j=1,\ldots,q\\
                    \label{eq:model-molecule-hc}
                    & H^{(j_1j_2)} = J_e \sigma^x_{j_1}\sigma^x_{j_2},\ (j_1, j_2)\in E\\
                    & \frac{d}{dt}X(t)=-iH(t)X(t)\nonumber\\
                    & X(0)= X_\textrm{init} \nonumber\\
                    \label{eq:model-circuit-cons-u-sum}
                    & \sum_{j=1}^{2q} u_j(t) + \sum_{(j_1,j_2)\in E}w_{j_1j_2}(t)=1\\
                    & u_j(t) \in \left\{0,1\right\},\ w_{j_1j_2}(t)\in \left\{0, 1\right\},\ j=1,\ldots,2q,\ (j_1, j_2)\in E,
\end{align}
\end{subequations}
where $J_c,\ J_f,\ J_e$ are the parameters corresponding to quantum machines and the subscript $(j)$ in constraints~\eqref{eq:model-molecule-h1} indicates that the matrix operation acts on the $j$th qubit. 
Constraint~\eqref{eq:model-circuit-cons-u-sum} enforces the SOS1 property for this problem. 
We choose $J_c=0.2\pi,\ J_f=3\pi,\ J_e=0.1\pi$ according to the paper by~\citet{gokhale2019partial}.
We test molecules H$_2$ (Dihydrogen), LiH (Lithium hydride), and BeH$_2$ (Beryllium dihydride) and generate the UCCSD circuits with minimum energy by the Python package Qiskit~\citep{aleksandrowicz2019qiskit}.  
The instances are labeled by CircuitH2, CircuitLiH, and CircuitBeH2, respectively. 

The settings of all the parameters are presented in Table~\ref{tab: setting}. 
We test parameters for penalizing the switching when rounding continuous controls $\alpha$ and $\beta$ both in the interval $[10^{-n},\ 10^{-n+1}]$ with step size $5\times 10^{-n-1}$ for $n=1,\ 2,\ 3,\ 4$. We choose the parameter with the best trade-off between objective values and the number of switches to present the results. 
All the numerical simulations were conducted on a macOS computer with 8 cores, 16GB RAM, and a 3.20GHz processor in Python 3.8. All the computational time results are the times on classical computers. Our full code and results are available on our GitHub repository~\citep{codeswitching2022}.
\begin{table}[htbp]
  \centering
  \caption{Parameter settings of examples. The parameters include the number of qubits ($q$), number of controllers ($N$), evolution time $t_f$, number of time steps ($T$), $L_2$ penalty parameter ($\rho$), and switching penalty parameter ($\alpha,\ \beta$). The $L_2$ penalty parameter $\rho$ is only for circuit compilation examples and it is marked by ``-'' for other examples.}
    \begin{tabular}{lrrrrrrr}
    \hline
    Instance & $q$ & $N$ & $t_f$ & $T$ & $\rho$ & $\alpha$ & $\beta$\\
    \hline
    Energy2 & 2     & 2  & 2   & 40 & - & 0.1 & 0.075 \\
    \hline
    Energy4 & 4     & 2  & 2    & 40  & - & 0.15 & 0.015\\
    \hline
    Energy6 & 6     & 2   & 5  & 100 & - & 0.015 & 0.01 \\
    \hline
    CNOT5 & 2     & 2   & 5   & 100 & - & 0.02 & 0.02 \\
    \hline
    CNOT10 & 2     & 2   & 10   & 200 & - & 0.003 & 0.008\\
    \hline
    CNOT20 & 2     & 2  & 20    & 400 & - & 0.01 & 0.015\\
    \hline
    NOT2 & 1 & 2 & 2 & 20 & - & 0.01 & 0.03\\
    \hline
    NOT6 & 1 & 2 & 6 & 60 & - & 0.0015 & 0.015\\
    \hline
    NOT10 & 1 & 2 & 10 & 100 & - & 0.009 & 0.035\\
    \hline
    CircuitH2 & 2     & 5   & 10   & 100 & 1.0 & 0.045 & 0.01\\
    \hline
    CircuitLiH & 4     & 12  & 20   & 200 & 0.1 & 0.03 & 0.06\\
    \hline
    CircuitBeH2 & 6     & 19 & 20 & 200 & 0.01 & 0.03 & 0.2\\
    \hline
    \end{tabular}%
  \label{tab: setting}%
\end{table}%

\subsection{Sensitivity Analysis of Switching Penalty}
\label{sec:sensitivity}
We take an instance Energy6 to show the performance of controls of Algorithm~\ref{alg:overall} with different switching penalty parameters $\alpha$ when obtaining binary controls based on objective values (Algorithm~\ref{alg:switch-eobj}). The performance of changing $\beta$ when obtaining binary controls based on the cumulative differences (Algorithm~\ref{alg:switch-cumdiff}) is similar. 

In the quantum energy minimization problem, the first excited state is the state having energy as the second smallest eigenvalue of the corresponding matrix $\Tilde{H}$. We use $E_{\textrm{fe}}$ to represent the energy of the first excited state. In most applications, we are interested in comparing obtained energy and $E_{\textrm{fe}}$. If the obtained energy is smaller, we say that the control is good enough to distinguish the states. In Table~\ref{tab:res-diff-alpha}, we present the difference between obtained energy and minimum energy $E_\textrm{min}$ for different penalty parameters as well as the difference between the first excited state energy and the minimum energy for the 5 randomly generated instances. 

\begin{table}[htbp]
  \centering
  \caption{Difference between obtained energy and minimum energy of Energy6 example for Algorithm~\ref{alg:overall} using Algorithm~\ref{alg:switch-eobj} at step~\ref{algline:round} with different switching penalty parameters $\alpha$. Column ``First-excited'' represents the difference between the first-excited energy and the ground energy. We bold the maximum parameter with obtained energy less than the first-excited energy for all 5 instances.}
    \begin{tabular}{lrrrrrrrr}
    \hline
     $\alpha$     & First-excited & 0     & 0.001 & 0.003 & 0.005  & 0.01  & \textbf{0.015} & 0.02 \\
    \hline
    Instance 1 & 0.8966 & 0.1503 & 0.1507 & 0.1518 & 0.1549 & 0.1825 & \textbf{0.2772} & 0.5358 \\
    Instance 2 & 0.9013 & 0.1235 & 0.1243 & 0.1289 & 0.1310 & 0.1993 & \textbf{0.2800} & 0.2800 \\
    Instance 3 & 0.8587 & 0.2156 & 0.2163 & 0.2194 & 0.2236 & 0.2313 & \textbf{0.2458} & 0.3283 \\
    Instance 4 & 1.4315 & 0.0356 & 0.0367 & 0.0409 & 0.0434 & 0.1639 & \textbf{0.1639} & 0.3532 \\
    Instance 5 & 0.2772 & 0.2059 & 0.2064 & 0.2082 & 0.2094 & 0.2200 & \textbf{0.2719} & 0.3969 \\
    Average & 0.8731 & 0.1462 & 0.1469 & 0.1499 & 0.1525 & 0.1994 & \textbf{0.2478} & 0.3788 \\
    \hline
     $\alpha$    & 0.03  & 0.04  & 0.05   & 0.07  & 0.1   & 0.2   & 0.3   & 0.6 \\
    \hline
    Instance 1 & 0.5358 & 0.5358 & 0.9666 & 0.9666 & 1.4327 & 1.5090 & 3.9625 & 3.9625 \\
    Instance 2 & 0.4363 & 0.7910 & 0.7910 & 0.7910 & 1.2370 & 3.2501 & 5.3743 & 6.2999 \\
    Instance 3 & 0.4504 & 0.6540 & 0.6540 & 0.6540 & 1.3813 & 2.4889 & 4.7132 & 4.7573 \\
    Instance 4 & 0.7195 & 0.7195 & 0.7195 & 1.6350 & 1.6350 & 2.4832 & 5.3054 & 5.3054 \\
    Instance 5 & 0.3969 & 0.6277 & 0.6277 & 0.9823 & 0.9823 & 1.5300 & 1.5300 & 4.1083 \\
    Average & 0.5078 & 0.6656 & 0.7518 & 1.0058 & 1.3337 & 2.2522 & 4.1771 & 4.8867 \\
    \hline
    \end{tabular}%
  \label{tab:res-diff-alpha}%
\end{table}%

We show that with the increase of switching penalty parameter $\alpha$, the differences of energy for all the instances increase, until reaching the maximum difference $-E_\textrm{min}$. We notice that when the difference between first-excited and minimum energy is small (Instance 5), we require a smaller energy ratio difference to distinguish the states, therefore requiring a smaller penalty parameter $\alpha$. We bold the maximum penalty parameter $\alpha=0.015$ with obtained energy smaller than the first-excited energy in all the instances.

In our model, obtaining energy less than the first-excited energy $E_{\textrm{fe}}$ is equivalent to obtaining an objective value smaller than $1-E_{\textrm{fe}}/E_{\textrm{min}}$. 
We present average objective values $1 - E/E_{\textrm{min}}$ and TV-norm values among $5$ instances for different switching penalty parameters in Table~\ref{tab:res-obj-tv-alpha}. 
The detailed objective values and TV-norm value results for each instance can be found in Appendix~\ref{app:results}. 
We show that with the increase of the switching penalty parameter $\alpha$, the average objective value increases monotonically and the average TV-norm value decreases monotonically.
\begin{table}[htbp]
  \centering
  \caption{Average objective value and TV-norm results of Energy6 example for Algorithm~\ref{alg:overall} using Algorithm~\ref{alg:switch-eobj} at step~\ref{algline:round} with different switching penalty parameters $\alpha$. TV-norm value of the first-excited state is not applicable and marked by ``-".}
    \begin{tabular}{lrrrrrrrr}
    \hline
    $\alpha$ & {First-excited} & 0     & 0.001 & 0.003 & 0.005  & 0.01  & 0.015 & 0.02\\
    \hline
    Objective & 0.1787 & 	0.0319 & 	0.0321 & 	0.0327 & 	0.0332 & 	0.0422 & 	0.0526 & 	0.0824 \\
    TV-norm & {-} & 116.4& 	74.0& 	41.2& 	32.4	& 23.2	& 20.4	& 17.6 \\
    \hline
    $\alpha$ &  0.03  & 0.04  & 0.05   & 0.07  & 0.1   & 0.2   & 0.3   & 0.6 \\
    \hline
    Objective & 0.1063&	0.1373&	0.1591&	0.2109&	0.2791&	0.4521&	0.8432&	1.0000 \\
    TV-norm & 15.2	&12.8	&12.0&	10.4&	8.4	&4.8	&2.4	&0.0 \\
    \hline
    \end{tabular}%
  \label{tab:res-obj-tv-alpha}%
\end{table}%
Furthermore, we present the trade-off relationship between objective values computed by $1-E/E_\textrm{min}$ and TV-norm values in Figure~\ref{fig:res-alpha}. 
Blue circles, orange triangles, green stars, red squares, and purple plus signs represent the results of various $\alpha$ of instances 1--5, respectively. We use larger transparent dots to represent the results with higher energy than first-excited state energy and use smaller opaque dots to represent the results with lower energy. We choose the point with energy lower than the first-excited state energy for both 5 instances and the smallest TV-norm as our best switching penalty parameter, which is $\alpha=0.015$. 
\begin{figure}[htbp]
    \centering
    \includegraphics[width=0.6\textwidth]{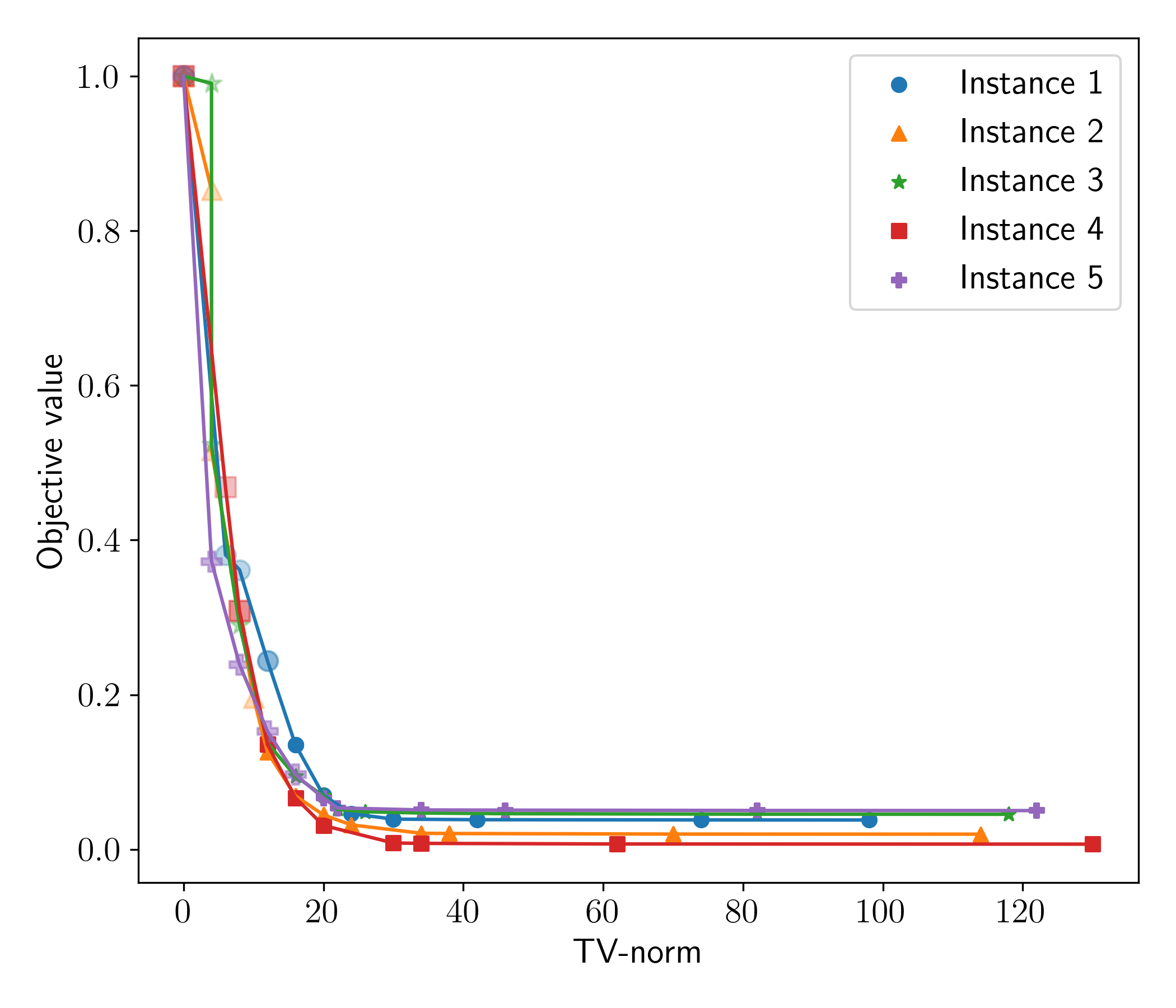}
    \caption{Objective value and TV-norm results of Energy6 example for Algorithm~\ref{alg:overall} using Algorithm~\ref{alg:switch-eobj} at step~\ref{algline:round} with different switching penalty parameters $\alpha$. Blue circles, orange triangles, green stars, red squares, and purple plus signs represent the results of instances 1--5, respectively. Larger transparent markers represent the results with higher energy than first-excited state energy while smaller opaque markers represent the results with lower energy.}
    \label{fig:res-alpha}
\end{figure}
\subsection{Numerical Results}
\label{sec:results}
In this section, we discuss the numerical results for four methods introduced in Section~\ref{sec:instance}. In all the following tables, figures and discussion, we label results of discretized controls obtained by~\citep[Algorithm TR+MT+ALB]{Fei2023binarycontrolpulse} as ``\citep{Fei2023binarycontrolpulse}''. We label the results obtained by Algorithm~\ref{alg:overall} using~\citep[Algorithm TR+MT+ALB]{Fei2023binarycontrolpulse}, heuristic methods based on objective values (Algorithm~\ref{alg:switch-eobj}), and heuristic methods based on cumulative difference (Algorithm~\ref{alg:switch-cumdiff}) to attain binary controls $u^\textrm{bin}$ at step~\ref{algline:round} as ``Alg.~\ref{alg:overall}w/\citep{Fei2023binarycontrolpulse}'', ``Alg.~\ref{alg:overall}w/\ref{alg:switch-eobj}'', and ``Alg.~\ref{alg:overall}w/\ref{alg:switch-cumdiff}'', respectively. 

\paragraph{Objective results}
We present the histograms of objective values with log-scale and TV-norm values in Figure~\ref{fig:all} for all the instances and methods. Blue bars represent objective value results and orange bars represent TV-norm results. Bars marked by slashes, backslashes, stars, and dots represent the results of methods TR+MT+ALB~\citep{Fei2023binarycontrolpulse}, Algorithm~\ref{alg:overall} with binary controls from TR+MT+ALB~\citep{Fei2023binarycontrolpulse}, Algorithm~\ref{alg:overall} with binary controls from Algorithm~\ref{alg:switch-eobj}, and Algorithm~\ref{alg:overall} with binary controls from Algorithm~\ref{alg:switch-cumdiff}, respectively.
\begin{figure}[htb]
    \centering
\includegraphics[width=\textwidth]{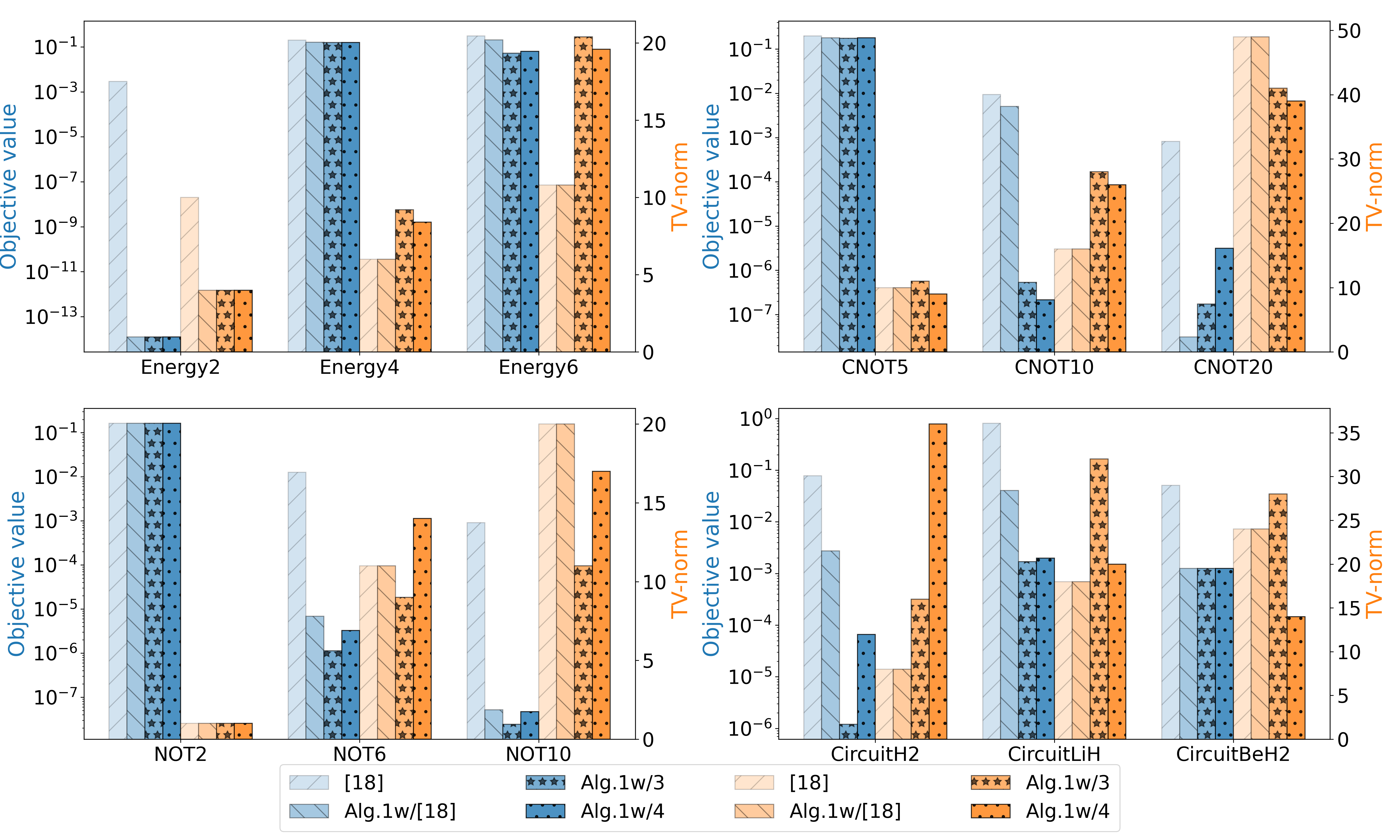}
\caption{Histograms of objective values with log-scale and TV-norm values for all the instances and methods. Blue bars represent objective value results and orange bars represent TV-norm results. Bars marked by slashes, backslashes, stars, and dots represent methods~\citep{Fei2023binarycontrolpulse}, Alg.~\ref{alg:overall} with binary controls from TR+MT+ALB~\citep{Fei2023binarycontrolpulse}, Alg.~\ref{alg:overall} with binary controls from Alg.~\ref{alg:switch-eobj}, and Alg.~\ref{alg:overall} with binary controls from Alg.~\ref{alg:switch-cumdiff}. }
    \label{fig:all}
\end{figure}

Comparing the direct results of~\citep{Fei2023binarycontrolpulse} and Algorithm~\ref{alg:overall} with binary controls from~\citep{Fei2023binarycontrolpulse}, we demonstrate that switching time optimization significantly improves the objective value because we can eliminate a fixed time discretization without increasing the number of switches. 
Compared to the results of Algorithm~\ref{alg:overall} with binary controls from~\citep{Fei2023binarycontrolpulse}, 
we show that Algorithm~\ref{alg:overall} with binary controls obtained from new heuristic methods (Algorithm~\ref{alg:switch-eobj}--\ref{alg:switch-cumdiff}) is not only more concise and convenient for adjusting parameters but also gains a better trade-off results for most instances. For the results of Algorithm~\ref{alg:overall} with our two new methods to obtain binary control, we illustrate that the objective values are in the same order but the TV-norm values vary among instances. In practice, we recommend testing two methods and choosing the best one. 

We select an instance NOT10 and present the control results with objective values (represented by ``Obj'') and TV-norm values in Figure~\ref{fig:ctrl-not}. We show that switching optimization improves the objective value significantly by slightly modifying the time points of switches by comparing the upper-left and upper-right figures in Figure~\ref{fig:ctrl-not}.
We demonstrate that switching time optimization with our new rounding algorithms reduces both objective values and TV-norm values in Figure~\ref{fig:ctrl-not}.

\paragraph{CPU time results} We present the CPU time of solving continuous relaxations (Step~\ref{algline:solvec}), extracting binary controls (Step~\ref{algline:round}), and conducting switching time optimization (Step~\ref{algline:st}) for all the methods in Table~\ref{tab:res-time}. 
We eliminate the time of merging intervals because it is not closely related to various instances and methods. 
We show our new methods of obtaining binary controls dramatically reduce the computational time compared to the previous method by eliminating the local search process. 
With our proposed acceleration technique in the heuristic method based on objective values, we avoid conducting time-evolution for $N\cdot T$ times and complete the process in seconds. 
The time of solving the switching optimization model~\eqref{eq:model-ts} is all less than one second. 

\begin{figure}[ht]
    \centering
\includegraphics[width=\textwidth]{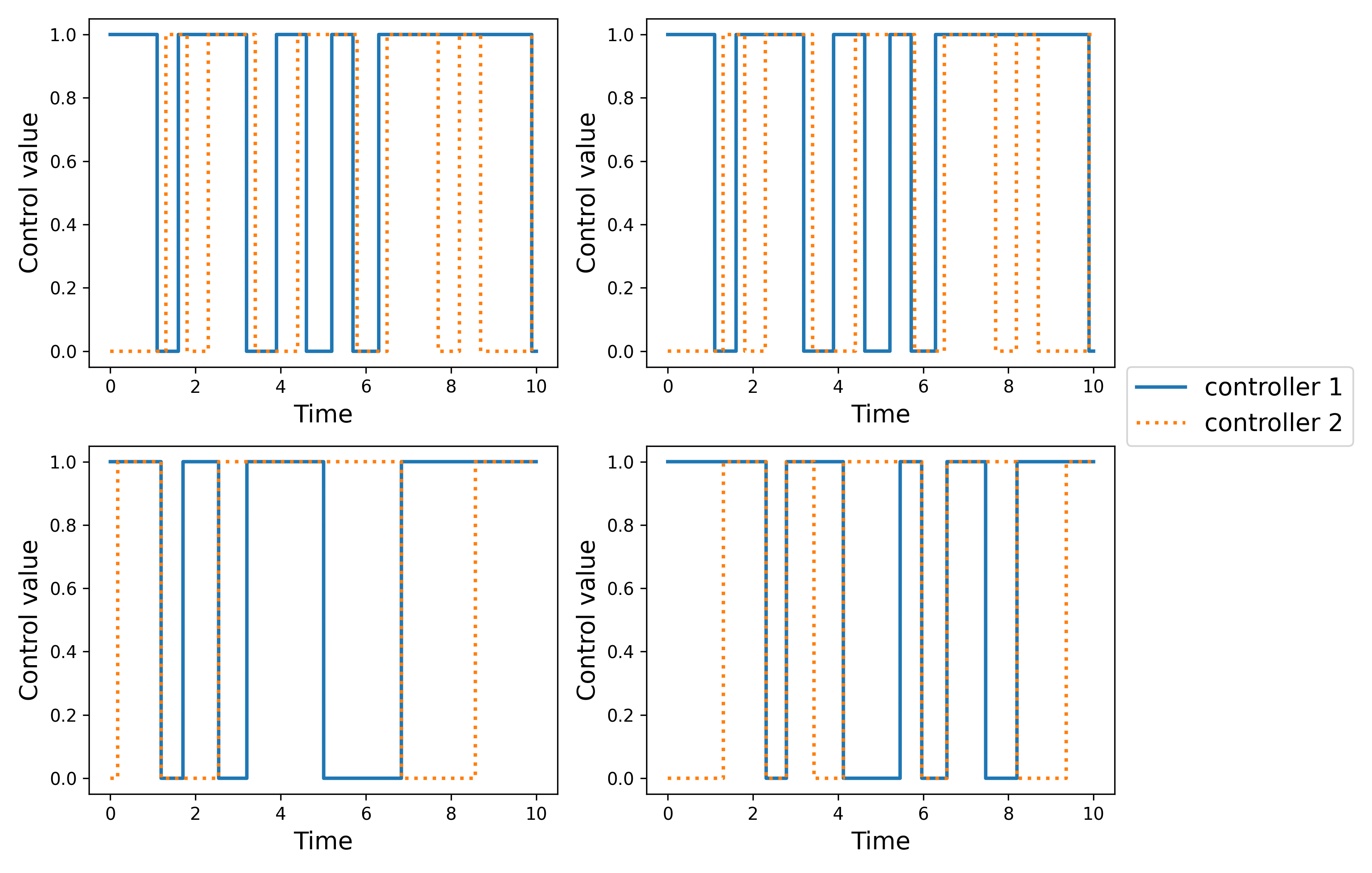}
    \caption{Control results of instance NOT10 for all the methods. 
    Upper-left: discretized controls obtained by~\citep{Fei2023binarycontrolpulse} (objective 9.087E$-$04, TV-norm 20). 
    Upper-right: controls of Algorithm~\ref{alg:overall} with binary controls obtained by~\citep{Fei2023binarycontrolpulse} (objective: 5.132E$-$08, TV-norm: 20). 
    Lower-left: controls of Algorithm~\ref{alg:overall} with binary controls obtained by Algorithm~\ref{alg:switch-eobj} (objective: 2.439E$-$08, TV-norm: 11). 
    Lower-right: controls of Algorithm~\ref{alg:overall} with binary controls obtained by Algorithm~\ref{alg:switch-cumdiff} (objective: 4.692E$-$08, TV-norm: 17). 
    Blue lines and orange dashed lines represent controllers 1 and 2, respectively.}
    \label{fig:ctrl-not}
\end{figure}
\begin{table}[ht]
  \centering
  \caption{CPU time (s) results of all the methods for main steps in Algorithm~\ref{alg:overall}. We eliminate the time of merging intervals. Column ``Continuous'' represents the results of solving continuous relaxations (Step~\ref{algline:solvec}). Columns ``Alg.~\ref{alg:overall}w/\citep{Fei2023binarycontrolpulse}'', ``Alg.~\ref{alg:overall}w/\ref{alg:switch-eobj}'', and ``Alg.~\ref{alg:overall}w/\ref{alg:switch-cumdiff}'' represent the time of rounding continuous controls (Step~\ref{algline:round}) and solving the switching time optimization model~\eqref{eq:model-ts} (Step~\ref{algline:st}) in Algorithm~\ref{alg:overall} using~\citep{Fei2023binarycontrolpulse}, Algorithm~\ref{alg:switch-eobj}, and Algorithm~\ref{alg:switch-cumdiff} in Step~\ref{algline:round}, respectively.}
  \begin{adjustbox}{width=\textwidth}
    \begin{tabular}{lr|rrr|rrr}
    \hline
          & Continuous
          & \multicolumn{3}{c|}{Binary control extraction (Step~\ref{algline:round})} & \multicolumn{3}{c}{Switching time optimization (Step~\ref{algline:st})} \\
          \cline{3-8}
          &   (Step~\ref{algline:solvec})    & {Alg.~\ref{alg:overall}w/\citep{Fei2023binarycontrolpulse}} & {Alg.~\ref{alg:overall}w/\ref{alg:switch-eobj}} & {Alg.~\ref{alg:overall}w/\ref{alg:switch-cumdiff}} & {Alg.~\ref{alg:overall}w/\citep{Fei2023binarycontrolpulse}} & {Alg.~\ref{alg:overall}w/\ref{alg:switch-eobj}} & {Alg.~\ref{alg:overall}w/\ref{alg:switch-cumdiff}} \\
    \hline
    Energy2 & 0.130 & 4.038 & 0.011 & 0.013 & 0.002 & 0.002 & 0.002 \\
    \hline
    Energy4 & 2.892 & 33.173 & 0.030 & 0.027 & 0.003 & 0.009 & 0.007 \\
    \hline
    Energy6 & 105.788 & 3174.384 & 0.678 & 0.484 & 0.063 & 0.463 & 0.167 \\
    \hline
    CNOT5 & 1.125 & 100.506 & 0.036 & 0.023 & 0.015 & 0.020 & 0.009 \\
    \hline
    CNOT10 & 0.725 & 291.447 & 0.071 & 0.043 & 0.035 & 0.036 & 0.044 \\
    \hline
    CNOT20 & 1.025 & 595.570 & 0.135 & 0.080 & 0.036 & 0.038 & 0.030 \\
    \hline
    NOT2  & 0.046 & 0.560 & 0.009 & 0.006 & 0.002 & 0.003 & 0.003 \\
    \hline
    NOT6  & 0.147 & 12.240 & 0.022 & 0.015 & 0.013 & 0.007 & 0.016 \\
    \hline
    NOT10 & 0.105 & 82.500 & 0.033 & 0.021 & 0.008 & 0.017 & 0.009 \\
    \hline
    CircuitH2 & 1.754 & 353.006 & 0.045 & 0.022 & 0.008 & 0.018 & 0.032 \\
    \hline
    CircuitLiH & 238.685 & 2347.433 & 0.392 & 0.131 & 0.015 & 0.042 & 0.025 \\
    \hline
    CircuitBeH2 & 3060.900 & 37368.162 & 2.026 & 1.758 & 0.124 & 0.555 & 0.311 \\
    \hline
    \end{tabular}%
    \end{adjustbox}
  \label{tab:res-time}%
\end{table}%

We present how the common logarithm of CPU time and iterations vary among problem sizes computed by $2^q\cdot T\cdot N$ when solving~\eqref{eq:model-ts} in Figure~\ref{fig:time}. 
We show that CPU times grow exponentially with $q$ because the dimensions of Hamiltonian matrices grow exponentially, while the number of iterations is stable. 
\begin{figure}[htbp]
    \centering
    \includegraphics[width=0.95\textwidth]{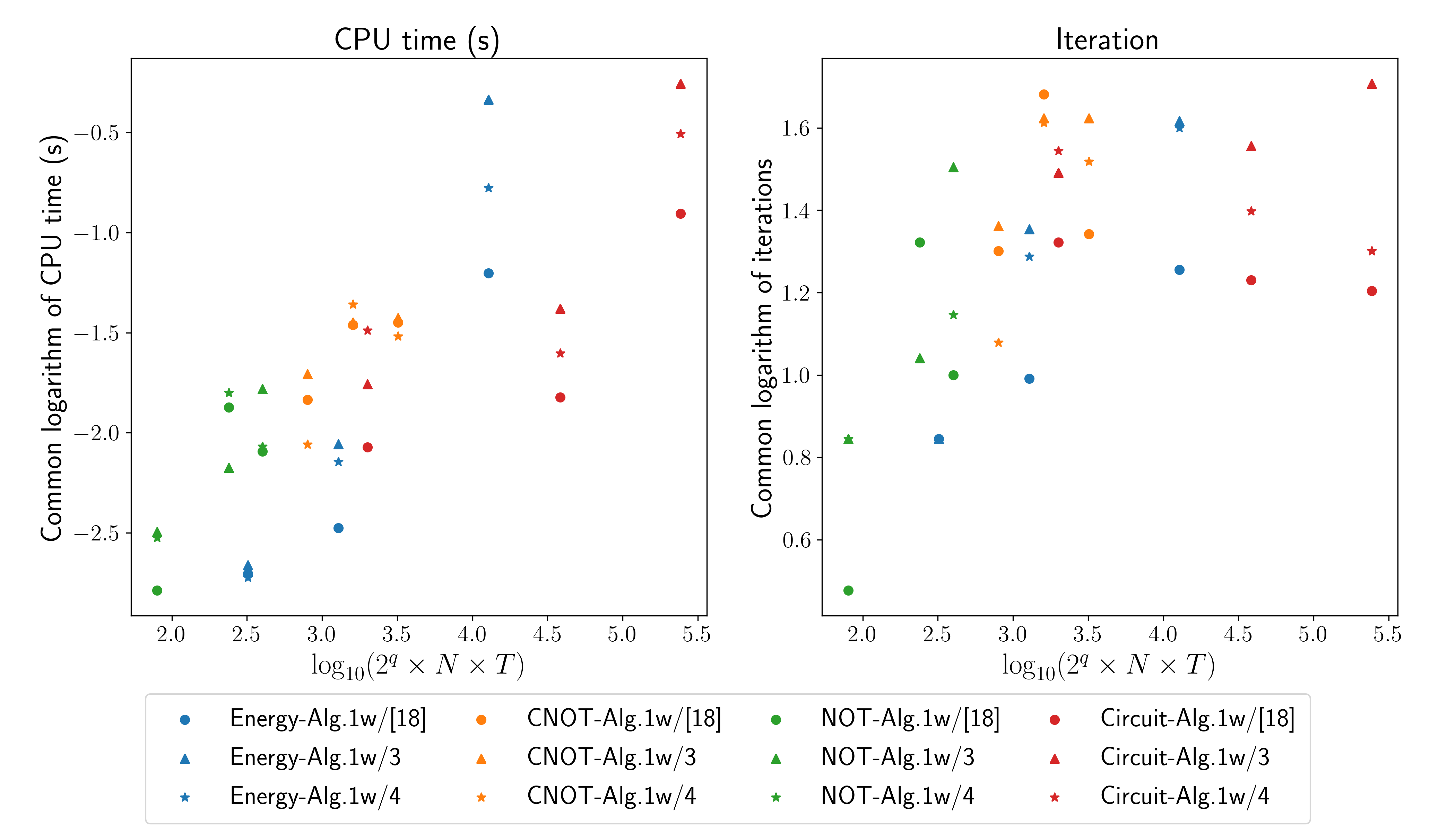}
    \caption{CPU times and numbers of iterations varying among problem sizes computed by $2^q\cdot T\cdot N$. We take the common logarithm of all the values. Blue, orange, green, and red dots represent instances of Energy, CNOT, NOT, and Circuit examples. Circles, triangles, and stars represent time results of Alg.~\ref{alg:overall} with~\citep{Fei2023binarycontrolpulse}, Alg.~\ref{alg:switch-eobj}, Alg.~\ref{alg:switch-cumdiff} to extract binary controls at Step~\ref{algline:round}.}
    \label{fig:time}
\end{figure}
The main CPU time of large-scale examples such as Energy6 and CircuitBeH2 comes from solving continuous relaxations. 
In practice, using quantum computers to conduct the time-evolution process can help save computational times with at most exponential reduction. We only need to evaluate the final state for computing the objective value and the gradient by the finite difference method. 

\paragraph{Acceleration of time-evolution simulation} We select two large-scale instances, Energy6 and CircuitBeH2 to indicate the benefits of our time-evolution simulation with pre-decomposed Hamiltonian controllers proposed in Section~\ref{sec:model-switching}.  
We present the CPU time of conducting time-evolution processes and the total CPU time when solving the model~\eqref{eq:model-ts} in Table~\ref{tab:res-time-acc}, where rows ``Accelerated'' represent the results of our accelerated simulation and rows ``Baseline'' represent the results of the baseline that computes matrix exponentials at each iteration. We demonstrate our accelerated simulation obtains remarkable improvement in CPU times with at most 16.3x speed-up.
\begin{table}[htbp]
  \centering
  \caption{CPU time (s) results of solving the switching time optimization model (Step~\ref{algline:st}) with extracted controllers by various methods, including the benchmark TR+MT+ALB~\citep{Fei2023binarycontrolpulse} in Column ``Alg.~\ref{alg:overall}w/\citep{Fei2023binarycontrolpulse}'', Algorithm~\ref{alg:switch-eobj} in Column ``Alg.~\ref{alg:overall}w/\ref{alg:switch-eobj}'', and Algorithm~\ref{alg:switch-cumdiff} in Column ``Alg.~\ref{alg:overall}w/\ref{alg:switch-cumdiff}''. We present CPU times of evolution simulations and total CPU times. Row ``\# Evolution'' presents the number of evolution times. Rows ``Acceleration'' and ``Baseline'' represent results after and before our acceleration by pre-decomposing matrices. Row ``Speed-up'' represents the speed-up results.}
  \begin{adjustbox}{width=\textwidth}
    \begin{tabular}{llrrrrrr}
    \hline
          &       & \multicolumn{3}{c}{Energy 6} & \multicolumn{3}{c}{CircuitBeH2} \\
    \cline{3-8}
    \multicolumn{2}{l}{} & Alg.~\ref{alg:overall}w/\citep{Fei2023binarycontrolpulse}  & Alg.~\ref{alg:overall}w/\ref{alg:switch-eobj}   & Alg.~\ref{alg:overall}w/\ref{alg:switch-cumdiff} & Alg.~\ref{alg:overall}w/\citep{Fei2023binarycontrolpulse}  & Alg.~\ref{alg:overall}w/\ref{alg:switch-eobj}  & Alg.~\ref{alg:overall}w/\ref{alg:switch-cumdiff} \\
    \hline
    \# Evolution && 18.0 & 41.4 & 39.8 & 16 & 51 & 20   \\
    \hline
    \multirow{3}[0]{*}{Evolution (s)} & Acceleration   & 0.035 & 0.255 & 0.095 & 0.052 & 0.235 & 0.163 \\
    \cline{2-8}
          & Baseline & 0.577 & 1.742 & 1.454 & 0.648 & 2.858 & 0.594 \\
    \cline{2-8}
          & Speed-up & 16.3x & 6.8x  & 15.3x & 12.6x & 12.1x & 12.5x \\
    \hline
    \multirow{3}[0]{*}{Total (s)} & Acceleration   & 0.063 & 0.463 & 0.167 & 0.124 & 0.555 & 0.311 \\
    \cline{2-8}
          & Baseline & 0.902 & 2.711 & 2.266 & 1.063 & 5.829 & 1.038 \\
    \cline{2-8}
          & Speed-up & 14.4x & 5.9x  & 13.6x & 8.6x  & 10.5x & 3.3x \\
    \hline
    \end{tabular}%
      \end{adjustbox}
  \label{tab:res-time-acc}%
\end{table}%

\paragraph{Results of fewer time steps} Because the switching time optimization model~\eqref{eq:model-ts} only requires the controller sequence and optimizes the time of all the switching points, we no longer require precise time discretization.
Therefore, we can solve continuous relaxation with longer discretized time intervals, thus reducing the number of time steps and computational costs. We take Energy6 and CircuitBeH2 as examples and solve instances with $T=20$ and $T=40$ by Algorithm~\ref{alg:overall} using our new methods, Algorithm~\ref{alg:switch-eobj}--\ref{alg:switch-cumdiff} to round controls. In Table~\ref{tab:res-fewer}, we compare objective values, TV-norm values, and CPU times between different numbers of time steps. We show that both methods obtain similar objective values and TV-norm values after switching time optimization with a significant decrease in CPU times.
\begin{table}[ht]
  \centering
  \caption{Comparison of objective value, TV-norm value, and CPU time results with different time steps on example Energy6 ($T=20,\ 100$) and  CircuitBeH2 ($T=40,\ 200$) for the switching time optimization model. Columns ``Alg.~\ref{alg:overall}w/\ref{alg:switch-eobj}'' and ``Alg.~\ref{alg:overall}w/\ref{alg:switch-cumdiff}'' represent results of Algorithm~\ref{alg:overall} with extracted binary controls obtained by Algorithm~\ref{alg:switch-eobj} and Algorithm~\ref{alg:switch-cumdiff}.}
  \begin{adjustbox}{width=\textwidth}
    \begin{tabular}{lrrrr|rrrr}
    \hline
          & \multicolumn{2}{c}{Energy6-T20} & \multicolumn{2}{c|}{Energy6-T100} & \multicolumn{2}{c}{CircuitBeH2-T40} & \multicolumn{2}{c}{CircuitBeH2-T200} \\
          \hline
          & {Alg.~\ref{alg:overall}w/\ref{alg:switch-eobj}} & {Alg.~\ref{alg:overall}w/\ref{alg:switch-cumdiff}} & {Alg.~\ref{alg:overall}w/\ref{alg:switch-eobj}} & {Alg.~\ref{alg:overall}w/\ref{alg:switch-cumdiff}} & {Alg.~\ref{alg:overall}w/\ref{alg:switch-eobj}} & {Alg.~\ref{alg:overall}w/\ref{alg:switch-cumdiff}} & {Alg.~\ref{alg:overall}w/\ref{alg:switch-eobj}} & {Alg.~\ref{alg:overall}w/\ref{alg:switch-cumdiff}} \\
          \hline
    Objective & 0.0723 & 0.0453 & 0.0526 & 0.0632 & 1.250E$-$03 & 1.249E$-$03 & 1.250E$-$03 & 1.250E$-$03 \\
    \hline
    TV-norm & 18.8  & 22  & 20.4  & 19.6  & 22    & 14    & 28    & 14 \\
    \hline
    CPU time & 38.579 & 38.609 & 106.929 & 106.439 & 329.643 & 329.383 & 3063.481 & 3062.387 \\
    \hline
    \end{tabular}%
    \end{adjustbox}
  \label{tab:res-fewer}%
\end{table}%

\subsection{Discussion of Controls}
\label{sec:control}
In this section, we choose Algorithm~\ref{alg:overall} with the method of obtaining binary controls based on objective values (Algorithm~\ref{alg:switch-eobj}) as an example to show the figures of controls and provide intuitive explanations of our overall framework. 
We select examples CircuitLiH, CircuitH2, and the fourth randomly generated instance of Energy6 to present the figures. 
For each example, we present the continuous controls $u^\textrm{con}$ after solving the continuous relaxation (Step~\ref{algline:solvec}), extracted binary controls by Algorithm~\ref{alg:switch-eobj} (Step~\ref{algline:round}), and controls after switching time optimization (Step~\ref{algline:st}) with their objective values. 

In Figure~\ref{fig:ctrl-lih}, we show the results of the example CircuitLiH. Comparing the objective value of binary controls and optimized control, we show that solving the switching time optimization model significantly improves the objective value. 
We notice that the binary controls obtained based on objective values have the most switches at the beginning of the evolution time interval because our heuristic algorithm aims to balance the objective value and the number of switches. 
Specifically, at the beginning of the time evolution interval, keeping the same control leads to a high objective value increase, so the system switches frequently among different controllers. At the second half of the time interval, the objective value becomes stable because it is upper-bounded by one, hence the algorithm chooses to keep the current controller to avoid an increase in the number of switches. 

\begin{figure}[htbp]
    \setcounter{subfigure}{0}
    \centering
    \includegraphics[width=\textwidth]{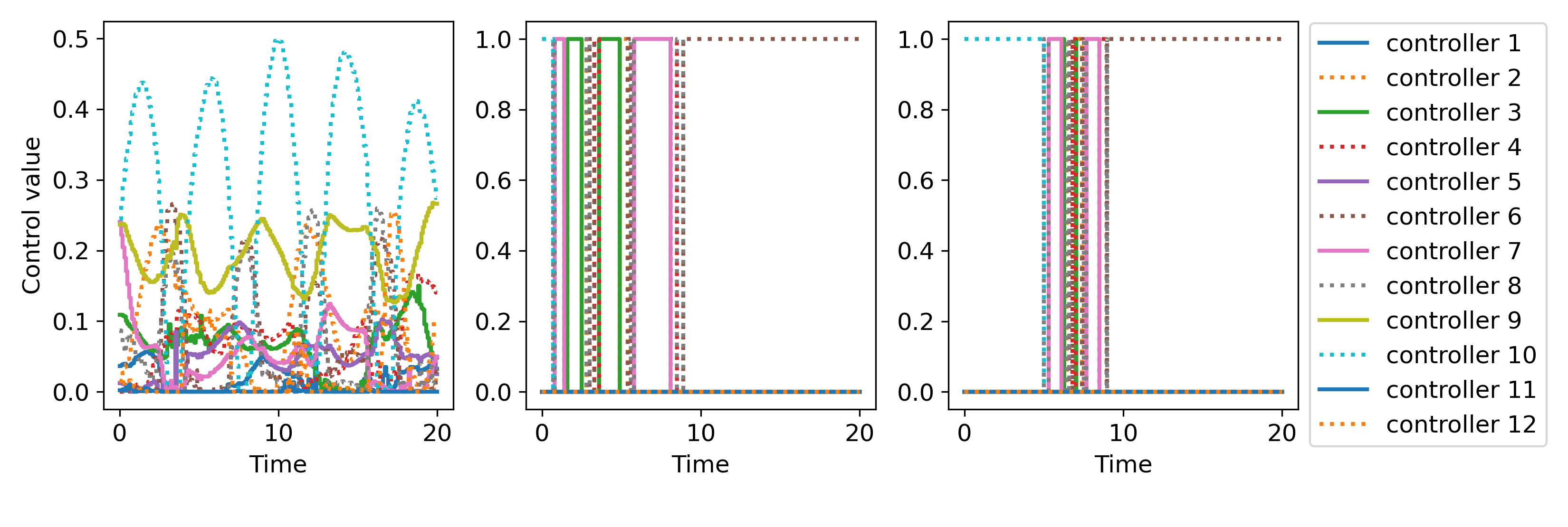}
    \caption{Control results for CircuitLiH of Algorithm~\ref{alg:overall} with binary controls obtained based on objective values (Algorithm~\ref{alg:switch-eobj}). 
    Left: Continuous controls with objective 1.310E$-$03. 
    Middle: Binary controls with objective 0.9993.
    Right: Optimized controls with objective 1.702E$-03$.}
    \label{fig:ctrl-lih}
\end{figure}

In Figure~\ref{fig:ctrl-h2}, we present the control results for the example CircuitH2. We demonstrate that the optimized control results obtain a smaller objective value even compared to the continuous controls, showing the advantages of the switching time optimization model by eliminating time discretization. Furthermore, we show that the switching time optimization model can obtain an optimal solution in that some time intervals have zero length, leading to a reduction in the number of switches.  
\begin{figure}[ht]
    \centering
    \includegraphics[width=\textwidth]{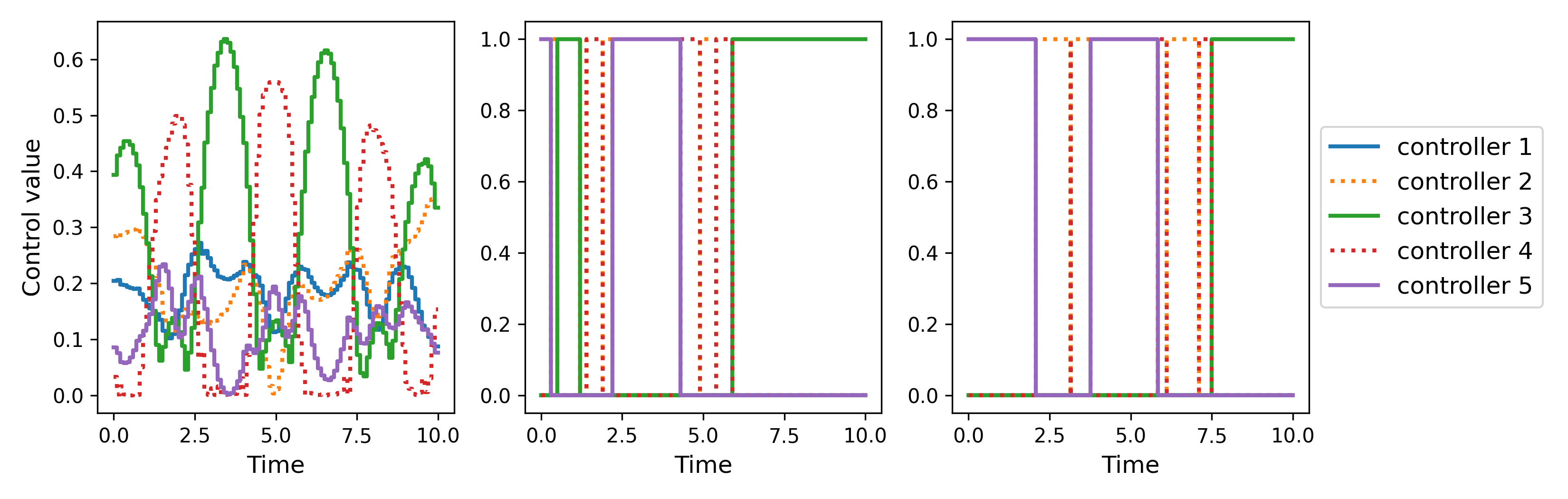}
    \caption{Control results for CircuitH2 of Algorithm~\ref{alg:overall} with binary controls obtained based on objective values (Algorithm~\ref{alg:switch-eobj}). 
    Left: Continuous controls with objective 2.150E$-$04. 
    Middle: Binary controls with objective 0.9619.
    Right: Optimized controls with objective 1.208E$-$06.}
    \label{fig:ctrl-h2}
\end{figure}

In Figures~\ref{fig:steps-obj-t100} and \ref{fig:steps-obj-t20}, we present the control results with larger time steps $T=100$ and fewer time steps $T=20$. We show that reducing the time steps leads to a higher objective value of continuous results and rounded binary controls, which is a disadvantage of the discretized model. 
However, our control results after the switching time optimization obtain the same objective for $T=100$ and $T=20$, indicating the capability of our model for reducing the computation of continuous relaxation with no degradation in the final optimized results. 
\begin{figure}[htbp]
    \centering
    \includegraphics[width=\textwidth]{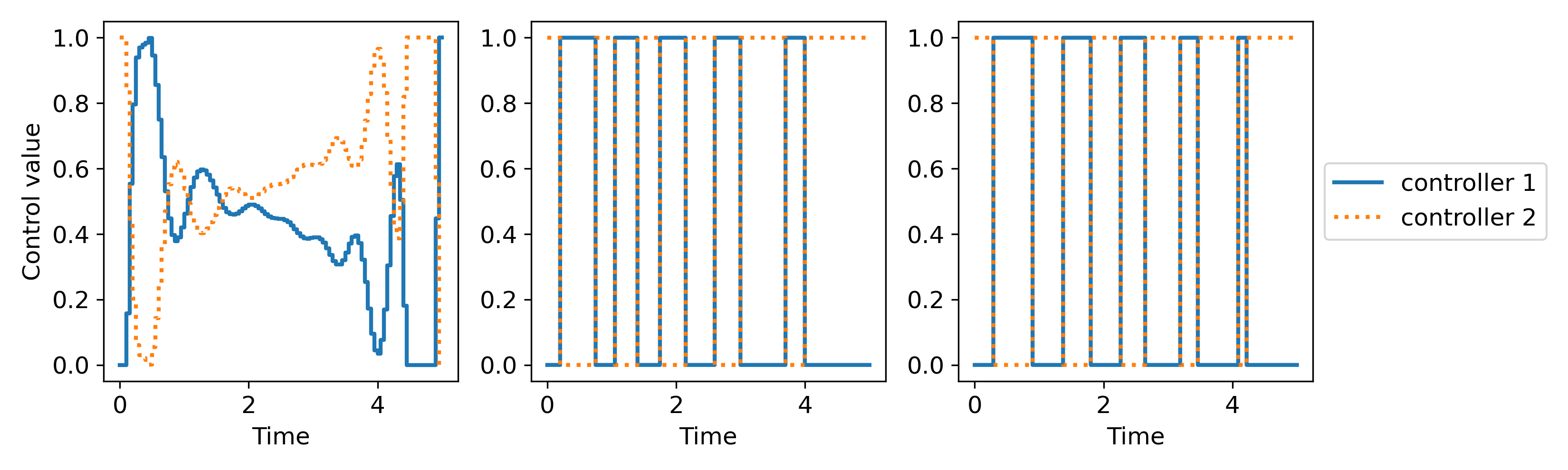}
    \caption{Control results for instance 4 of Energy6 of Algorithm~\ref{alg:overall} with binary controls obtained based on objective values (Algorithm~\ref{alg:switch-eobj}). The number of time steps $T=100$.
    Left: Continuous controls with objective 0.0067. 
    Middle: Binary controls with objective 0.1808.
    Right: Optimized controls with objective 0.0309.}
    \label{fig:steps-obj-t100}
\end{figure}
\begin{figure}[ht]
    \centering
    \includegraphics[width=\textwidth]{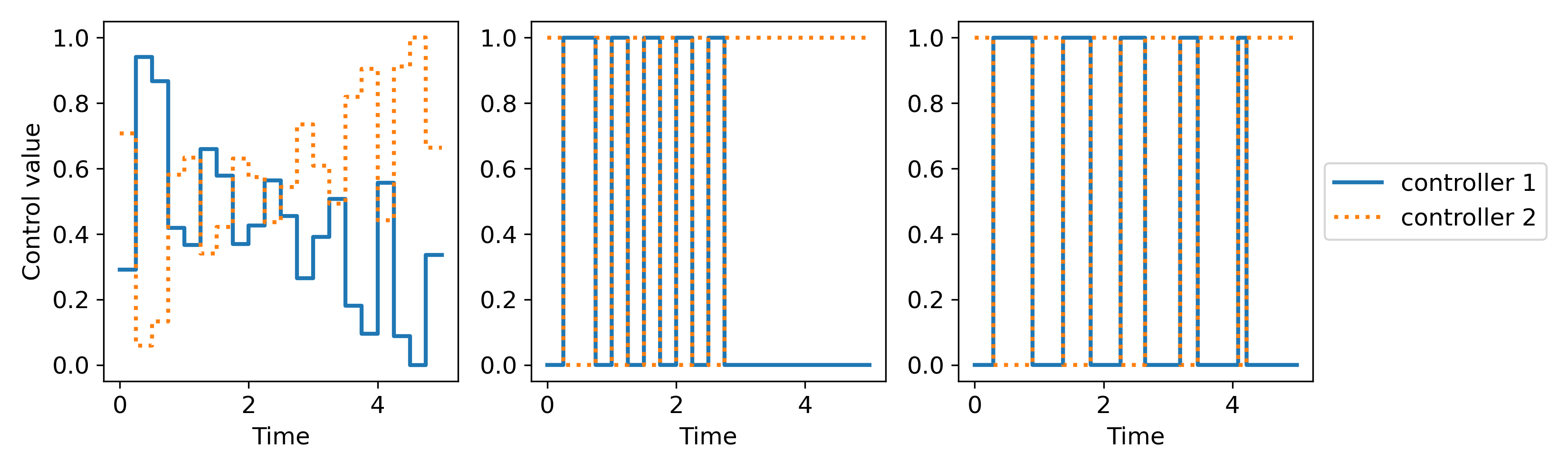}
    \caption{Control results for instance 4 of Energy6 of Algorithm~\ref{alg:overall} with binary controls obtained based on objective values (Algorithm~\ref{alg:switch-eobj}). The number of time steps $T=20$.
    Left: Continuous controls with objective 0.0078. 
    Middle: Binary controls with objective 0.2231.
    Right: Optimized controls with objective 0.0309.}
    \label{fig:steps-obj-t20}
\end{figure}

%% file: chpt5_conclusion.tex
\section{Conclusion}
\label{sec:conclusion}
In this paper, we studied a generic binary quantum control problem. 
We proposed an algorithmic framework optimizing control functions as well as time points of switching controllers. 
We proposed two heuristic methods to obtain controller sequences from discretized continuous controls with a penalty on the number of switches. 
With obtained controller sequences, we developed a switching time optimization model to optimize the switching points of controllers and applied a sequential least-squares programming algorithm to solve it. 
Furthermore, we proposed an acceleration technique to compute matrix exponentials for the time-evolution simulation on classical computers. 

With numerical experiments on multiple instances in different quantum systems, we demonstrated that our new framework significantly outperforms binary solutions of the time discretization model in balancing objective values and switching frequency. 
Although approaches to obtaining binary controls are various, we showed that our new heuristic methods are concise and effective, and obtain the best quality control after switching time optimization. In practice, one can test new heuristic methods and adjust penalty parameters to achieve the best control. 
Furthermore, we indicated that the switching time optimization model requires less precise time discretization, leading to a reduction of the computational burden. 

Because our algorithm only requires the final state of quantum systems to evaluate the objective value and the gradient, using quantum computers to conduct the time-evolution process is one of our future research directions. 
Noise and uncertainty in quantum systems have been a long-standing problem, making it a valuable topic to take account into uncertainties in quantum control problems and design robust controls. 

%% file: appendix.tex
\section{Proofs of All Theorems in Section~\ref{sec:alg-extraction}}
\label{app:proof}
\begin{proof}[Proof of Lemma~\ref{lemma:singular-4}]
Let $A=\sum_{i=1}^m w_i\sigma_i(A)v_i^T$ be the singular value decomposition of $A$ where $w_1,\ldots,w_m$ and $v_1,\ldots,v_m$ are orthonormal vectors. Then the norm of the trace 
\begin{align}
    \left|\mathbf{tr}\left\{UA\right\}\right| & = \left|\mathbf{tr}\left\{ \sum_{i=1}^m Uw_i\sigma_i(A) v_i^\dagger\right\}\right|
    \leq \sum_{i=1}^m \left|\mathbf{tr}\left\{Uw_i\sigma_i(A) v_i^\dagger\right\}\right|\nonumber\\
    & \leq \sum_{i=1}^m \|Uw_i\|_2\sigma_i(A)\|v_i\|_2\leq \sum_{i=1}^m \sigma_i(A)\leq m\sigma_1(A).
\end{align}
The first and second inequalities follow from the norm inequality. The third inequality holds because $Uw_1,\ldots,Uw_m$ and $v_1,\ldots,v_m$ are orthonormal vectors. The last inequality comes from the definition of singular values. 
\end{proof}

\begin{proof}[Proof of Theorem~\ref{theo:str}]
At time step $k$, for each examined control $u^{(j)}$ with only controller $j$ active, we define the change of objective value as $R_{jk}=\Bar{F}(u^{(j)}) - F_\textrm{cur}$.
We prove the theorem by showing that at each time step, there exists a controller such that the change of objective value $R_{jk}\leq 2C_1e^{C_2\Delta t} \Delta t^2$. 
Our proof mainly consists of two parts. Firstly we prove that the sum of all the terms in each $R_{jk}$ corresponding to $\Delta t^l,\ l\geq 2$ is upper bounded by $2C_1e^{C_2\Delta t} \Delta t^2$. 
Then we prove that there exists a controller $\hat{j}$ such that the sum of the constant terms and terms corresponding to $\Delta t$ in $R_{\hat{j}k}$ is upper bounded by $0$. 
{Before we present our detailed proof, we define a constant parameter
\begin{align}
    \sigma_{\textrm{max}} = \max_{j=1,\ldots,N} \sigma_1(H^{(j)}) + \sigma_1(H^{(0)}),
\end{align}
where $\sigma_1(\cdot )$ is the maximum singular value of a given constant Hamiltonian controller. 
}
We discuss the change of objective value for the aforementioned two specific objective function formulations{~\eqref{eq:obj-energy}--\eqref{eq:obj-fid}}. 
We start from the infidelity function because the formulation is more concise to prove. 
\paragraph{Infidelity function}
For each time step $k$, the change of objective value $R_{jk}$ is computed as 
\begin{align}
    R_{jk} = & \frac{1}{2^q}\left|\mathbf{tr}\left\{X_{\textrm{targ}}^\dagger \prod_{l=k+1}^{T} U_l e^{-iH^{(0)} \Delta t-i\sum_{j'=1}^N u^\textrm{con}_{j'k}H^{(j')}\Delta t} \prod_{l=1}^{k-1} U_l X_\textrm{init}\right\}\right|\nonumber\\
    & - \frac{1}{2^q}\left|\mathbf{tr}\left\{X_{\textrm{targ}}^\dagger \prod_{l=k+1}^{T} U_l e^{-iH^{(0)}\Delta t -iH^{(j)}\Delta t} \prod_{l=1}^{k-1} U_l X_\textrm{init}\right\}\right|,
\end{align}
where 
\begin{align}
    U_l = \begin{cases}
    e^{-iH^{(0)}\Delta t -i\sum_{j'=1}^N u^\textrm{bin}_{j'l} H^{(j)} \Delta t} & 1\leq l\leq k - 1\\
    e^{-iH^{(0)}\Delta t -i\sum_{j'=1}^N u^\textrm{con}_{j'l} H^{(j)} \Delta t} & k + 1\leq l\leq T
    \end{cases}.
\end{align}
For simplicity, we denote 
\begin{align}
    \mathbf{U}^{(k)} = \prod_{l=1}^{k-1} U_l X_\textrm{init}X_{\textrm{targ}}^\dagger \prod_{l=k+1}^{T} U_l, 
\end{align}
which is still a unitary matrix because it is the product of a series of unitary matrices. 
Following the definition of trace, we know that the trace of a product of matrices is invariant under cyclic permutations. Hence, 
\begin{align}
    R_{jk} = \frac{1}{2^q}\left|\mathbf{tr}\left\{\mathbf{U}^{(k)} e^{-iH^{(0)}\Delta t -i\sum_{j'=1}^N u^\textrm{con}_{j'k}H^{(j')}\Delta t}\right\}\right| - \frac{1}{2^q}\left|\mathbf{tr}\left\{\mathbf{U}^{(k)} e^{-iH^{(0)}\Delta t -iH^{(j)}\Delta t}\right\}\right| .
\end{align}
Without loss of generality, we consider a general matrix $A\in \mathbb{C}^{m\times m}$, then the exponential matrix $\displaystyle e^{A\Delta t} = \sum_{l=0}^\infty \frac{1}{l!} A^l\Delta t^l$. Therefore we have
\begin{align}
    \label{eq:trace-any-matrix}
    \mathbf{tr}\left\{\mathbf{U}^{(k)} e^{A\Delta t}\right\} = \mathbf{tr}\left\{\mathbf{U}^{(k)} (I + A\Delta t)\right\} + \sum_{l=2}^\infty \frac{1}{l!} \mathbf{tr}\left\{\mathbf{U}^{(k)} A^l\right\}\Delta t^l.
\end{align}
The norm of the second term of the trace in~\eqref{eq:trace-any-matrix} is upper bounded as
\begin{align}
    \left|\sum_{l=2}^\infty \frac{1}{l!} \mathbf{tr}\left\{\mathbf{U}^{(k)} A^l\right\}\Delta t^l\right|
     & \leq \sum_{l=2}^\infty \frac{m}{l!} \sigma_1(A^l)\Delta t^l\leq
     \sum_{l=2}^\infty \frac{m}{l!} (\sigma_1(A)\Delta t)^l 
     \leq m\sigma_1^2(A)e^{\sigma_1(A)\Delta t}\Delta t^2.
\end{align}
The first inequality follows from Lemma~\ref{lemma:singular-4} and the second inequality follows from Lemma~\ref{lemma:singular-2}. 
The last inequality directly comes from the Taylor expansion. 
{We define constants $C_1(A),\ C_2(A)$ corresponding to matrix $A$ as
\begin{align}
    & C_1(A) = m\sigma^2_1(A),\quad C_2(A) = \sigma_1(A),
\end{align}
then the second term is upper bounded by $C_1(A)e^{C_2(A)\Delta t}\Delta t^2$.}
Therefore we have the following inequality for the norm of trace~\eqref{eq:trace-any-matrix}: 
\begin{align}
    \label{eq:trace-simplify}
    \left|\mathbf{tr}\left\{\mathbf{U}^{(k)} (I + A\Delta t)\right\} \right| -  
    C_1(A)e^{C_2(A)\Delta t}\Delta t^2 & \leq \left|\mathbf{tr}\left\{\mathbf{U}^{(k)} e^{A\Delta t}\right\}\right| \nonumber\\
    & \leq \left|\mathbf{tr}\left\{\mathbf{U}^{(k)} (I + A\Delta t)\right\} \right|+
    C_1(A)e^{C_2(A)\Delta t}\Delta t^2, 
\end{align}
In our examples, all the Hamiltonian control matrices have a formulation $A=-i\left(^{(0)} + \sum_{j=1}^N u_j H^{(j)}\right)$ with given controls $\sum_{j=1}^N u_j=1$ and given constant controllers $H^{(0)},\ H^{(j)}\in \mathbb{C}^{2^q\times 2^q}$. Then the maximum singular value of any control matrix
\begin{align}
    \label{eq:singular-val}
    \sigma_1(A) \leq \sigma_1(H^{(0)}) + \sum_{j=1}^N u_j \sigma_1(H^{(j)})\leq \sigma_1(H^{(0)}) + \max_{j=1,\ldots,N} \sigma_1(H^{(j)}) = \sigma_\textrm{max}.
\end{align}
We define constants $C_1,\ C_2$ as  
\begin{align}
    \label{eq:constant-ub}
    \quad 
    C_1 = \sigma_\textrm{max}^2,\quad 
    C_2= \sigma_\textrm{max}.
\end{align}
For any {time-dependent Hamiltonian} matrix $A$, the constants $C_1(A)\leq 2^qC_1,\ C_2(A)\leq C_2$.
Substituting $\displaystyle A=-iH^{(0)}-i\sum_{j'=1}^N u^\textrm{con}_{j'k}H^{(j')}$ and $\displaystyle A=-iH^{(0)}-iH^{(j)}$ into the inequality~\eqref{eq:trace-simplify} and combining~\eqref{eq:constant-ub}, the change of objective value satisfies
\begin{align}
    R_{jk} \leq &
    \frac{1}{2^q}\left|\mathbf{tr}\left\{\mathbf{U}^{(k)} \left(I-iH^{(0)}\Delta t -i\sum_{j'=1}^N u^\textrm{con}_{j'k}H^{(j')}\Delta t\right)\right\}\right| \nonumber\\
    & - \frac{1}{2^q}\left|\mathbf{tr}\left\{\mathbf{U}^{(k)} \left(I-iH^{(0)}\Delta t -iH^{(j)}\Delta t\right)\right\}\right|
    + 2C_1e^{C_2\Delta t} \Delta t^2.
\end{align}
We prove that there exists a controller $j$ such that the value of the first two terms is no more than $0$ by contradiction. Assume that for any $j$, 
\begin{align}
    \label{eq:assumption}
    \left|\mathbf{tr}\left\{\mathbf{U}^{(k)} \left(I-iH^{(0)}\Delta t -i\sum_{j'=1}^N u^\textrm{con}_{j'k}H^{(j')}\Delta t\right)\right\}\right| - \left|\mathbf{tr}\left\{\mathbf{U}^{(k)} \left(I-iH^{(0)}\Delta t -iH^{(j)}\Delta t\right)\right\}\right| > 0.
\end{align}
Then we have 
\begin{align}
    \left|\mathbf{tr}\left\{\mathbf{U}^{(k)} \left(I-iH^{(0)}\Delta t -i\sum_{j'=1}^N u^\textrm{con}_{j'k}H^{(j')}\Delta t\right)\right\}\right|& \leq \sum_{j'=1}^N u^\textrm{con}_{j'k}\left|\mathbf{tr}\left\{\mathbf{U}^{(k)} \left(I-iH^{(0)}\Delta t -i H^{(j')}\Delta t\right)\right\}\right|\nonumber\\
    & < \left|\mathbf{tr}\left\{\mathbf{U}^{(k)} \left(I-iH^{(0)}\Delta t -i\sum_{j'=1}^N u^\textrm{con}_{j'k}H^{(j')}\Delta t\right)\right\}\right|.
\end{align}
The first inequality follows from $\sum_{j'=1}^N u^\textrm{con}_{j'k}=1$ and $0\leq u^\textrm{con}_{j'k}\leq 1$. The second inequality follows from~\eqref{eq:assumption}. It obviously leads to a contradiction. Therefore we prove that there exists at least one controller $j$ such that 
\begin{align}
    R_{jk}\leq 2C_1e^{C_2\Delta t} \Delta t^2.
\end{align}
In the algorithm, at each time step, we choose the $\hat{j}$ with minimum change of objective value. Taking the summation over all time steps, we have
\begin{align}
    \bar{F}(u^\textrm{bin})-\bar{F}(u^\textrm{con})=\sum_{k=1}^T \min_{j=1,\ldots,N} R_{jk} \leq 2C_1e^{C_2\Delta t} \Delta t.
\end{align}

\paragraph{Energy function} We have a similar proof for the energy function. For each time step $k$, the change of objective value 
$R_{jk}$ is computed as
\begin{align}
    R_{jk} = & \langle \psi_0| \left(\mathbf{U}^{(2k)} e^{-iH^{(0)}\Delta t -i\sum_{j'=1}^N u^\textrm{con}_{j'k} H^{(j')}\Delta t} \mathbf{U}^{(1k)}\right)^\dagger \tilde{H} \mathbf{U}^{(2k)} e^{-iH^{(0)}\Delta t -i\sum_{j'=1}^N u^\textrm{con}_{j'k} H^{(j')}\Delta t} \mathbf{U}^{(1k)}|\psi_0 \rangle / E_{\textrm{min}}\nonumber\\
    & - \langle \psi_0| \left(\mathbf{U}^{(2k)} e^{-iH^{(0)}\Delta t -iH^{(j)}\Delta t} \mathbf{U}^{(1k)}\right)^\dagger \tilde{H} \mathbf{U}^{(2k)} e^{-iH^{(0)}\Delta t -i H^{(j)}\Delta t} \mathbf{U}^{(1k)} |\psi_0 \rangle / E_{\textrm{min}}, 
\end{align}
where 
\begin{align}
    \mathbf{U}^{(1k)} = \prod_{l=1}^{k-1} e^{-iH^{(0)}\Delta t -i\sum_{j'=1}^N u^\textrm{bin}_{j'l} H^{(j')} \Delta t}X_\textrm{init},\quad 
    \mathbf{U}^{(2k)} = \prod_{l=k+1}^T e^{-iH^{(0)}\Delta t -i\sum_{j'=1}^N u^\textrm{con}_{j'l} H^{(j')} \Delta t}.
\end{align}
It is obvious that $\mathbf{U}^{(1k)}, \mathbf{U}^{(2k)}$ are both unitary matrices. We consider a general Hamiltonian matrix $A\in \mathbb{C}^{m\times m}$, with the matrix exponential expansion, we have
\begin{subequations}
\label{eq:obj-exp-energy}
\begin{align}
    & \langle \psi_0| \left(\mathbf{U}^{(2k)} e^{A\Delta t} \mathbf{U}^{(1k)}\right)^\dagger \tilde{H} \mathbf{U}^{(2k)} e^{A\Delta t} \mathbf{U}^{(1k)}|\psi_0 \rangle \nonumber \\
    =& \langle\psi_0| \left(\mathbf{U}^{(2k)} \mathbf{U}^{(1k)}\right)^\dagger \tilde{H} \mathbf{U}^{(2k)}  \mathbf{U}^{(1k)}|\psi_0 \rangle \nonumber \\
    & + \langle \psi_0| \left(\mathbf{U}^{(2k)} A \mathbf{U}^{(1k)}\right)^\dagger \tilde{H} \mathbf{U}^{(2k)}  \mathbf{U}^{(1k)}|\psi_0 \rangle \Delta t + \langle \psi_0| \left(\mathbf{U}^{(2k)} \mathbf{U}^{(1k)}\right)^\dagger \tilde{H} \mathbf{U}^{(2k)} A \mathbf{U}^{(1k)}|\psi_0 \rangle \Delta t\nonumber\\
    \label{eq:square-1}
    & + \langle \psi_0| \left(\mathbf{U}^{(2k)} A \mathbf{U}^{(1k)}\right)^\dagger \tilde{H} \mathbf{U}^{(2k)} A  \mathbf{U}^{(1k)}|\psi_0 \rangle \Delta t^2\\
    \label{eq:square-2}
    & + \langle \psi_0| \left(\mathbf{U}^{(2k)} \sum_{l=2}^\infty \frac{1}{l!}A^l\Delta t^l \mathbf{U}^{(1k)}\right)^\dagger \tilde{H} \mathbf{U}^{(2k)} e^{A\Delta t} \mathbf{U}^{(1k)}|\psi_0 \rangle \\
    \label{eq:square-3}
    & + \langle \psi_0| \left(\mathbf{U}^{(2k)} \mathbf{U}^{(1k)}\right)^\dagger \tilde{H} \mathbf{U}^{(2k)} \sum_{l=2}^\infty\frac{1}{l!} A^l\Delta t^l \mathbf{U}^{(1k)}|\psi_0 \rangle \\
    \label{eq:square-4}
    & + \langle \psi_0| \left(\mathbf{U}^{(2k)} A \Delta t \mathbf{U}^{(1k)}\right)^\dagger \tilde{H} \mathbf{U}^{(2k)} \sum_{l=2}^\infty \frac{1}{l!}A^l\Delta t^l \mathbf{U}^{(1k)}|\psi_0 \rangle.
\end{align}
\end{subequations}
{We first prove that there exists constants $C_1(A),\ C_2(A)$ such that the summation of terms~\eqref{eq:square-1}--\eqref{eq:square-4} is bounded by $C_1(A)e^{C_2(A)\Delta t}\Delta t^2$.
}
Because $\mathbf{U}^{(1k)}$ and $\mathbf{U}^{(2k)}$ are unitary matrices and $\||\psi_0\rangle\|_2=1$, following from Lemma~\ref{lemma:singular-1} and Lemma~\ref{lemma:singular-3}, the norm value of term~\eqref{eq:square-1} is no more than 
\begin{align}
    \sigma_1(A^\dagger \tilde{H}A)\Delta t^2\leq \sigma_1^2(A) \sigma_1(\tilde{H})\Delta t^2\leq \sigma_1^2(A)  \sigma_1(\tilde{H})e^{\sigma_1(A)\Delta t}\Delta t^2.
\end{align}
The first inequality follows from Lemma~\ref{lemma:singular-2}.
The second inequality holds because $\Delta t$ and $\sigma_1(A)$ are non-negative.
Recall that $e^{A\Delta t}$ is a unitary matrix because $A$ is a Hamiltonian matrix, hence by Lemma~\ref{lemma:singular-1} and Lemma~\ref{lemma:singular-3}, the norm value of term~\eqref{eq:square-2} is no more than 
\begin{align}
    \sum_{l=2}^\infty \frac{1}{l!} \sigma_1(A^l\tilde{H})\Delta t^l
    \leq \sigma_1(\tilde{H}) \sum_{l=2}^\infty \frac{1}{l!} \sigma_1^l(A) \Delta t^l
    \leq \sigma_1(\tilde{H}) e^{\sigma_1(A)\Delta t} \Delta t^2.
\end{align}
The first inequality follows from Lemma~\ref{lemma:singular-2} and the second inequality follows from Taylor expansion. 
Similarly, the norm value of term~\eqref{eq:square-3} is bounded by 
\begin{align}
    \sigma_1(\tilde{H}) e^{\sigma_1(A)\Delta t} \Delta t^2, 
\end{align}
and the norm value of term~\eqref{eq:square-4} is bounded by 
\begin{align}
    \sigma_1(A^\dagger \tilde{H}) e^{\sigma_1(A)\Delta t} \Delta t^3\leq t_f \sigma_1(A) \sigma_1(\tilde{H}) e^{\sigma_1(A)\Delta t} \Delta t^2,
\end{align}
where the inequality holds because of the fact that $\Delta t\leq t_f$ and Lemma~\ref{lemma:singular-2}. 
We define the constants $C_1(A),\ C_2(A)$ as
\begin{subequations}
    \begin{align}
    & C_1(A) = \sigma_1(\tilde{H}) \left(\sigma_1^2(A) + 2 + t_f\sigma_1(A)\right)\\
    & C_2(A) = \sigma_1(A).
    \end{align}
\end{subequations}
The values of {the summation of terms}~\eqref{eq:square-1}--\eqref{eq:square-4} {is} upper bounded by the summation of their norms, which is bounded by $C_1(A)e^{C_2(A)\Delta t}\Delta t^2$. 
We define the constants $C_1,\ C_2$ as 
\begin{subequations}
\begin{align}
    & C_1 = \sigma_1(\tilde{H}) \left(\sigma^2_\textrm{max} + 2 + t_f\sigma_\textrm{max}\right) / |E_\textrm{min}|\\
    & C_2 = \sigma_\textrm{max}.
\end{align}
\end{subequations}
From~\eqref{eq:singular-val}, we know that for any {time-dependent} Hamiltonian controller $A$, we have $C_1(A)\leq |E_\textrm{min}| C_1,\ C_2(A)\leq C_2$. Substituting $\displaystyle A=-iH^{(0)}-i\sum_{j'=1}^N u^\textrm{con}_{j'k}H^{(j')}$ and $\displaystyle A=-iH^{(0)}-iH^{(j)}$ into the expanded trace equation~\eqref{eq:obj-exp-energy} and combining with the constants $C_1,\ C_2$, the change of objective value satisfies
\begin{align}
    R_{jk}\leq & \frac{1}{E_{\textrm{min}}}\langle \psi_0| \left(\mathbf{U}^{(2k)} i\left(H^{(j)} - \sum_{j'=1}^N u^\textrm{con}_{j'k}H^{(j')}\right) \mathbf{U}^{(1k)}\right)^\dagger \tilde{H} \mathbf{U}^{(2k)} \mathbf{U}^{(1k)}|\psi_0 \rangle \Delta t \nonumber\\
    & + \frac{1}{E_{\textrm{min}}} \langle \psi_0| \left(\mathbf{U}^{(2k)} \mathbf{U}^{(1k)}\right)^\dagger  \tilde{H} \mathbf{U}^{(2k)} i \left(H^{(j)} - \sum_{j'=1}^N u^\textrm{con}_{j'k}H^{(j')}\right) \mathbf{U}^{(1k)}|\psi_0 \rangle \Delta t \nonumber\\
    & + 2C_1e^{C_2\Delta t}\Delta t^2. 
\end{align}
We prove that there exists a controller $j$ such that the value of terms of $\Delta t$ is upper bounded by zero. Assume that there does not exist such a controller, then by taking summation weighted by $u^\textrm{con}_{jk}$ over all the controllers, we have
\begin{align}
\label{eq:ineq-energy}
    0 < & \frac{1}{E_{\textrm{min}}}\langle \psi_0| \left(\mathbf{U}^{(2k)} i\sum_{j=1}^N u^\textrm{con}_{jk}\left(H^{(j)} - \sum_{j'=1}^N u^\textrm{con}_{j'k}H^{(j')}\right) \mathbf{U}^{(1k)}\right)^\dagger \tilde{H} \mathbf{U}^{(2k)} \mathbf{U}^{(1k)}|\psi_0 \rangle \Delta t \nonumber\\
    & + \frac{1}{E_{\textrm{min}}} \langle \psi_0| \left(\mathbf{U}^{(2k)} \mathbf{U}^{(1k)}\right)^\dagger  \tilde{H} \mathbf{U}^{(2k)} i \sum_{j=1}^N u^\textrm{con}_{jk} \left(H^{(j)} - \sum_{j'=1}^N u^\textrm{con}_{j'k}H^{(j')}\right) \mathbf{U}^{(1k)}|\psi_0 \rangle \Delta t. 
\end{align}
Because $\sum_{j=1}^N u^\textrm{con}_{jk}=1$, we have 
\begin{align}
    \sum_{j=1}^N u^\textrm{con}_{jk} \left(H^{(j)} - \sum_{j'=1}^N u^\textrm{con}_{j'k}H^{(j')}\right)
     = \sum_{j=1}^N u^\textrm{con}_{jk} H^{(j)} - \sum_{j'=1}^N u^\textrm{con}_{j'k}H^{(j')} = 0.
\end{align}
Therefore the right-hand side of the inequality~\eqref{eq:ineq-energy} is $0$, which leads to a contradiction. 
At each time step $k$, by our update rule in the algorithm, the change of objective value 
\begin{align}
    R_{\hat{j}k}=\min_{j=1,\ldots,N} R_{jk} \leq 2C_1e^{C_2\Delta t}\Delta t^2.
\end{align}
Summing over all time steps, we have
\begin{align}
    \bar{F}(u^\textrm{bin})-\bar{F}(u^\textrm{con})=\sum_{k=1}^T \min_{j=1,\ldots,N} R_{jk} \leq 2C_1e^{C_2\Delta t}\Delta t.
\end{align}
\end{proof}

\begin{proof}[Proof of Theorem~\ref{theo:str-no-sos1}]
We first define the constants $C_1,\ C_2$. 
Similar to~\eqref{eq:singular-val}, for any time-dependent Hamiltonian matrix $A = -i\left(H^{(0)}+ \sum_{j=1}^N u_jH^{(j)}\right)$, 
we now have $\sum_{j=1}^N u_j\leq N$ because $u_j\in [0,1],\ j=1,\ldots,N$. Hence the maximum singular value
\begin{align}
\label{eq:hamil-no-sos1}
    \sigma_1(A)\leq \sigma_1(H^{(0)}) + \sum_{j=1}^N u_j \sigma_1(H^{(j)})\leq N \sigma_\textrm{max}.
\end{align}
For energy objective function~\eqref{eq:obj-energy}, we define constants $C_1,\ C_2$ are specified as
\begin{subequations}
    \begin{align}
    & C_1 = \sigma_1(\tilde{H}) \left(N^2\sigma^2_\textrm{max} + 2 + t_fN\sigma_\textrm{max}\right) / |E_\textrm{min}|\\
    & C_2 = N\sigma_\textrm{max}.
\end{align}
For infidelity objective function~\eqref{eq:obj-fid}, $C_1,\  C_2$ are specified as
\begin{align}
    C_1 = N^2\sigma^2_\textrm{max},\quad 
    C_2= N\sigma_\textrm{max}.
\end{align}
\end{subequations}
With the same proof in Theorem~\ref{theo:str}, we show that the summation of terms with order $O(\Delta t^n),\ n\geq 2$ in $R_{jk}$ 
is upper bounded by $C_1e^{C_2\Delta t}\Delta t^2$.
We modify the proof of Theorem~\ref{theo:str} to prove that the terms with respect to $O(\Delta t)$ in $R_{jk}$ is upper bounded by $C_0\epsilon^c(\Delta t)$ where $C_0$ is a constant parameter. 
For simplicity of the proof, we define $\epsilon_k$ as 
\begin{align}
    \epsilon_k = \left|\sum_{j=1}^N u^\textrm{con}_{jk} - 1\right|,\ k=1,\ldots,T.
\end{align}
We discuss the proof specifically for two objective functions following the sequence in the proof of Theorem~\ref{theo:str}. 

\paragraph{Infidelity function}
We define the constant $C_0=\sigma_\textrm{max}$ and prove the conclusion by contradiction. Assume that for any controller $j$, it holds that
\begin{align}
    &\frac{1}{2^q}\left|\mathbf{tr}\left\{\mathbf{U}^{(k)} \left(I-iH^{(0)}\Delta t -i\sum_{j'=1}^N u^\textrm{con}_{j'k}H^{(j')}\Delta t\right)\right\}\right|\nonumber\\
    \label{eq:assumption-i-no-sos1}
    -& \frac{1}{2^q}\left|\mathbf{tr}\left\{\mathbf{U}^{(k)} \left(I-iH^{(0)}\Delta t-iH^{(j)}\Delta t\right)\right\}\right| > C_0\epsilon_k\Delta t.
\end{align}
Now we have 
\begin{subequations}
\begin{align}
    & 2^q C_0\epsilon_k\Delta t =  2^q \sum_{j=1}^N \frac{u^\textrm{con}_{jk}}{\sum_{j'=1}^N u^\textrm{con}_{j'k}} C_0\epsilon_k \Delta t\\
    \label{ieq:1}
    < & \left|\mathbf{tr} \left\{\mathbf{U}^{(k)} \left(I-iH^{(0)}\Delta t-i\sum_{j=1}^N u^\textrm{con}_{jk}H^{(j)}\Delta t\right)\right\}\right|\nonumber\\
    & - \sum_{j=1}^N \frac{u^\textrm{con}_{jk}}{\sum_{j'=1}^N u^\textrm{con}_{j'k}} \left|\mathbf{tr}\left\{\mathbf{U}^{(k)} \left(I-i H^{(0)} \Delta t -i H^{(j)}\Delta t\right)\right\}\right|\\
    \label{ieq:2}
    \leq & \left|\mathbf{tr}\left\{\mathbf{U}^{(k)} \left(-i\sum_{j=1}^N u^\textrm{con}_{jk} H^{(j)}\Delta t + i\sum_{j=1}^N \frac{u^\textrm{con}_{jk}}{\sum_{j'=1}^N u^\textrm{con}_{j'k}} H^{(j)}\Delta t\right) \right\}\right|\\
    \label{ieq:3}
    \leq & \frac{\left|\sum_{j'=1}^N {u^\textrm{con}_{j'k}} - 1\right|}{\sum_{j'=1}^N {u^\textrm{con}_{j'k}}}\left|\mathbf{tr}\left\{\mathbf{U}^{(k)}\left(-i\sum_{j=1}^N u^\textrm{con}_{jk}H^{(j)}\Delta t\right)\right\}\right|\\
    \leq & \frac{\epsilon_k}{\sum_{j'=1}^N u^\textrm{con}_{j'k}} 2^q \sum_{j=1}^N u^\textrm{con}_{jk} \sigma_\textrm{max}\Delta t
    = 2^q \sigma_\textrm{max} \epsilon_k \Delta t = 2^q C_0\epsilon_k\Delta .
    \end{align}
\end{subequations}
The inequality~\eqref{ieq:1} comes from the weighted summation of~\eqref{eq:assumption-i-no-sos1} over all the controllers with weight $u_{jk}/\sum_{j'=1}^N u_{j'k}$ for controller $j=1,\ldots,N$. The inequalities~\eqref{ieq:2}--\eqref{ieq:3} follow from the norm inequality and the fact that the summation of all the weights is one. The last inequality follows from the definition of $\epsilon_k$ and Lemma~\ref{lemma:singular-4}. By substituting $C_0$, the inequalities lead to a contradiction, which means that at each step, the minimum change of objective value 
\begin{align}
    \min_{j=1,\ldots,N} R_{jk}\leq 2C_1e^{C_2\Delta t}\Delta t^2 + C_0\epsilon_k\Delta t. 
\end{align}
Taking summation over all the time steps, we have 
\begin{align}
    \bar{F}(u^\textrm{bin}) - \bar{F}(u^\textrm{con}) = \sum_{k=1}^T \min_{j=1,\ldots,N} R_{jk}\leq 2C_1e^{C_2 \Delta t}\Delta t + C_0\epsilon^c(\Delta t).
\end{align}

\paragraph{Energy function}
We define $C_0 = 2  \sigma_1(\tilde{H})\sigma_\textrm{max}/|E_{\textrm{min}}|$ and prove the statement by contradiction. Assume that for any controller $j$, it holds that 
\begin{align}
\label{eq:assumption-e-no-sos1}
    & \frac{1}{E_{\textrm{min}}}\langle \psi_0| \left(\mathbf{U}^{(2k)} i\left(H^{(j)} - \sum_{j'=1}^N u^\textrm{con}_{j'k}H^{(j')}\right) \mathbf{U}^{(1k)}\right)^\dagger \tilde{H} \mathbf{U}^{(2k)} \mathbf{U}^{(1k)}|\psi_0 \rangle \Delta t \nonumber\\
    & + \frac{1}{E_{\textrm{min}}} \langle \psi_0| \left(\mathbf{U}^{(2k)} \mathbf{U}^{(1k)}\right)^\dagger  \tilde{H} \mathbf{U}^{(2k)} i \left(H^{(j)} - \sum_{j'=1}^N u^\textrm{con}_{j'k}H^{(j')}\right) \mathbf{U}^{(1k)}|\psi_0 \rangle \Delta t > C_0 \epsilon_k \Delta t.
\end{align}
Then we have 
\begin{subequations}
\begin{align}
    & C_0 \epsilon_k \Delta t = \sum_{j=1}^N \frac{u^\textrm{con}_{jk}}{\sum_{j'=1}^N u^\textrm{con}_{j'k}} 
     C_0\epsilon_k \Delta t\\
     \label{eq:ieq-e-1}
    < & \frac{\sum_{j=1}^N u^\textrm{con}_{jk}}{E_{\textrm{min}}}\langle \psi_0| \left(\mathbf{U}^{(2k)} i\frac{1-\sum_{j'=1}^N u^\textrm{con}_{j'k}}{\sum_{j'=1}^N u^\textrm{con}_{j'k}}H^{(j)} \mathbf{U}^{(1k)}\right)^\dagger \tilde{H} \mathbf{U}^{(2k)} \mathbf{U}^{(1k)}|\psi_0 \rangle \Delta t \nonumber\\
    & + \frac{\sum_{j=1}^N u^\textrm{con}_{jk}}{E_{\textrm{min}}} \langle \psi_0| \left(\mathbf{U}^{(2k)} \mathbf{U}^{(1k)}\right)^\dagger  \tilde{H} \mathbf{U}^{(2k)} i \frac{1-\sum_{j'=1}^N u^\textrm{con}_{j'k}}{\sum_{j'=1}^N u^\textrm{con}_{j'k}}H^{(j)} \mathbf{U}^{(1k)}|\psi_0 \rangle \Delta t\\
    \label{eq:ieq-e-2}
    \leq & \frac{2 }{|E_{\textrm{min}}|\sum_{j'=1}^N u^\textrm{con}_{j'k}} \epsilon_k   \sum_{j=1}^N u_{jk}^\textrm{con} \sigma_\textrm{max}\sigma_1(\tilde{H})\Delta t
    = \frac{2}{|E_{\textrm{min}}|} \sigma_{\textrm{max}}\sigma_1(\tilde{H}) \epsilon_k \Delta t=C_0 \epsilon_k \Delta t.
\end{align}
\end{subequations}
The inequality~\eqref{eq:ieq-e-1} comes from the weighted summation of~\eqref{eq:assumption-e-no-sos1} over all the controllers with weight $u_{jk}/\sum_{j=1}^N u_{jk}$ for controller $j=1,\ldots,N$. Notice that we use the fact that the summation of all the weights is one. The last inequality~\eqref{eq:ieq-e-2} follows from the definition of $\epsilon_k$ and Lemma~\ref{lemma:singular-4}. By substituting $C_0$, the inequalities lead to a contradiction, which means that at each step, the minimum change of objective value 
\begin{align}
    \min_{j=1,\ldots,N} R_{jk}\leq 2C_1e^{C_2 \Delta t}\Delta t^2 + C_0\epsilon_k\Delta t.
\end{align}
Taking summation over all the time steps, we have 
\begin{align}
    \bar{F}(u^\textrm{bin}) - \bar{F}(u^\textrm{con}) = \sum_{k=1}^T \min_{j=1,\ldots,N} R_{jk}\leq 2C_1e^{C_2 \Delta t}\Delta t + C_0\epsilon^c(\Delta t).
\end{align}
\end{proof}

\section{Detailed Numerical Results}
\label{app:results}
\paragraph{Detailed Results for Sensitivity Analysis}
For the sensitivity analysis in Section~\ref{sec:sensitivity}, 
we present the detailed objective values for $5$ instances in Table~\ref{tab:res-obj-alpha}. 
\begin{table}[htbp]
  \centering
  \caption{Objective values for Energy6 example solved by Algorithm~\ref{alg:overall} using Algorithm~\ref{alg:switch-eobj} to round continuous controls with various switching penalty parameter $\alpha$.}
    % \begin{tabular}{lrrrrrrrr}
    % \hline
    % $\alpha$ & {First-excited} & 0     & 0.001 & 0.003 & 0.005 & 0.01  & 0.02 \\
    % \hline
    % Instance 1 & 0.226 & 0.038 & 0.038 & 0.038 & 0.039 & 0.040 & 0.070 \\
    % Instance 2 & 0.143 & 0.020 & 0.020 & 0.020 & 0.020 & 0.022 & 0.044 \\
    % Instance 3 & 0.180 & 0.045 & 0.045 & 0.046 & 0.047 & 0.049 & 0.069 \\
    % Instance 4 & 0.270 & 0.007 & 0.007 & 0.008 & 0.008 & 0.009 & 0.031 \\
    % Instance 5 & 0.067 & 0.050 & 0.050 & 0.051 & 0.051 & 0.054 & 0.066 \\
    % Average & 0.177 & 0.032 & 0.032 & 0.033 & 0.033 & 0.035 & 0.056 \\
    % \hline
    % $\alpha$ & 0.03  & 0.04  & 0.05  & 0.08  & 0.1   & 0.3   & 0.6 \\
    % \hline
    % Instance 1 & 0.135 & 0.244 & 0.244 & 0.244 & 0.503 & 1.000 & 1.000 \\
    % Instance 2 & 0.066 & 0.122 & 0.125 & 0.249 & 0.196 & 0.837 & 1.000 \\
    % Instance 3 & 0.069 & 0.095 & 0.137 & 0.358 & 0.290 & 0.523 & 1.000 \\
    % Instance 4 & 0.067 & 0.135 & 0.136 & 0.308 & 0.308 & 0.918 & 1.000 \\
    % Instance 5 & 0.097 & 0.153 & 0.153 & 0.239 & 0.239 & 0.900 & 1.000 \\
    % Average & 0.087 & 0.150 & 0.159 & 0.280 & 0.307 & 0.836 & 1.000 \\
    % \hline
    % \end{tabular}%
    % Table generated by Excel2LaTeX from sheet 'new_alpha'
    \begin{tabular}{lrrrrrrrr}
    \hline
          & First-excited &  0     & 0.001 & 0.003 & 0.005  & 0.01  & 0.015 & 0.02 \\
    \hline
    Instance 1 & 0.2263 & 0.0379 & 0.0380 & 0.0383 & 0.0391 & 0.0461 & 0.0700 & 0.1352 \\
    Instance 2 & 0.1431 & 0.0196 & 0.0197 & 0.0205 & 0.0208 & 0.0316 & 0.0444 & 0.0444 \\
    Instance 3 & 0.1805 & 0.0453 & 0.0455 & 0.0461 & 0.0470 & 0.0486 & 0.0517 & 0.0690 \\
    Instance 4 & 0.2698 & 0.0067 & 0.0069 & 0.0077 & 0.0082 & 0.0309 & 0.0309 & 0.0666 \\
    Instance 5 & 0.0675 & 0.0501 & 0.0502 & 0.0507 & 0.0510 & 0.0536 & 0.0662 & 0.0966 \\
    Average & 0.1787 & 0.0319 & 0.0321 & 0.0327 & 0.0332 & 0.0422 & 0.0526 & 0.0824 \\
    \hline
         & 0.03  & 0.04  & 0.05   & 0.07  & 0.1   & 0.2   & 0.3   & 0.6  \\
    \hline
    Instance 1 & 0.1352 & 0.1352 & 0.2439 & 0.2439 & 0.3616 & 0.3808 & 1.0000 & 1.0000 \\
    Instance 2 & 0.0693 & 0.1256 & 0.1256 & 0.1256 & 0.1963 & 0.5159 & 0.8531 & 1.0000 \\
    Instance 3 & 0.0947 & 0.1375 & 0.1375 & 0.1375 & 0.2904 & 0.5232 & 0.9907 & 1.0000 \\
    Instance 4 & 0.1356 & 0.1356 & 0.1356 & 0.3082 & 0.3082 & 0.4680 & 1.0000 & 1.0000 \\
    Instance 5 & 0.0966 & 0.1528 & 0.1528 & 0.2391 & 0.2391 & 0.3724 & 0.3724 & 1.0000 \\
    Average & 0.1063 & 0.1373 & 0.1591 & 0.2109 & 0.2791 & 0.4521 & 0.8432 & 1.0000 \\
    \hline
    \end{tabular}%
  \label{tab:res-obj-alpha}%
\end{table}%
We present the detailed TV-norm values for $5$ instances in Table~\ref{tab:res-tv-alpha}. 
\begin{table}[htbp]
  \centering
  \caption{TV-norm results for Energy6 example solved by Algorithm~\ref{alg:overall} using Algorithm~\ref{alg:switch-eobj} to round continuous controls with various switching penalty parameter $\alpha$. TV-norm values of the first-excited state are marked by "-".}
    \begin{tabular}{lrrrrrrrr}
    \hline
          & First-excited & 0     & 0.001 & 0.003 & 0.005  & 0.01  & 0.015 & 0.02 \\
    \hline
    Instance 1 & -     & 98    & 74    & 42    & 30    & 24    & 20    & 16 \\
    Instance 2 & -     & 114   & 70    & 38    & 34    & 24    & 20    & 20 \\
    Instance 3 & -     & 118   & 82    & 46    & 34    & 26    & 22    & 20 \\
    Instance 4 & -     & 130   & 62    & 34    & 30    & 20    & 20    & 16 \\
    Instance 5 & -     & 122   & 82    & 46    & 34    & 22    & 20    & 16 \\
    Average & -     & 116.4 & 74.0  & 41.2  & 32.4  & 23.2  & 20.4  & 17.6 \\
    \hline
          & 0.03  & 0.04  & 0.05   & 0.07  & 0.1   & 0.2   & 0.3   & 0.6  \\
    \hline
    Instance 1 & 16    & 16    & 12    & 12    & 8     & 6     & 0     & 0 \\
    Instance 2 & 16    & 12    & 12    & 12    & 10    & 4     & 4     & 0 \\
    Instance 3 & 16    & 12    & 12    & 12    & 8     & 4     & 4     & 0 \\
    Instance 4 & 12    & 12    & 12    & 8     & 8     & 6     & 0     & 0 \\
    Instance 5 & 16    & 12    & 12    & 8     & 8     & 4     & 4     & 0 \\
    Average & 15.2  & 12.8  & 12.0  & 10.4  & 8.4   & 4.8   & 2.4   & 0.0 \\
    \hline
    \end{tabular}%
  \label{tab:res-tv-alpha}%
\end{table}%

\paragraph{Detailed Results for Objective}
We present the detailed numerical results for objective values and TV-norm values for all the methods and instances in Table~\ref{tab:res-obj}. 
We present the best discretized control obtained by~\citep[Algorithm TR+MT+ALB]{Fei2023binarycontrolpulse} in column ``~\citep{Fei2023binarycontrolpulse}'' as a baseline. 
% Columns ``{ST(Disc\citep{Fei2023binarycontrolpulse})}'', ``{ST(Obj[Alg.~\ref{alg:switch-eobj}])}'', and ``{ST(Cdiff[Alg.~\ref{alg:switch-cumdiff}])}'' represent results 
Columns ``Alg.~\ref{alg:overall}w/\citep{Fei2023binarycontrolpulse}'', ``Alg.~\ref{alg:overall}w/\ref{alg:switch-eobj}'', and ``Alg.~\ref{alg:overall}w/\ref{alg:switch-cumdiff}'' represent results 
obtained by Algorithm~\ref{alg:overall} with binary control obtained by ~\citep[Algorithm TR+MT+ALB]{Fei2023binarycontrolpulse}, the heuristic method based on objective value (Algorithm~\ref{alg:switch-eobj}), and the heuristic method based on the cumulative difference (Algorithm~\ref{alg:switch-cumdiff}), respectively. 
\begin{table}[htbp]
  \centering
  \caption{Objective and TV-norm value results of various approaches. 
  Column ``\citep{Fei2023binarycontrolpulse}'' represent the results of the benchmark discretized controls without optimizing switching points. 
  Columns ``Alg.~\ref{alg:overall}w/\citep{Fei2023binarycontrolpulse}'', ``Alg.~\ref{alg:overall}w/\ref{alg:switch-eobj}'', and ``Alg.~\ref{alg:overall}w/\ref{alg:switch-cumdiff}'' represent the results of Algorithm~\ref{alg:overall} with extracting binary controls by various approaches, including TR+MT+ALB~\citep{Fei2023binarycontrolpulse}, Algorithm~\ref{alg:switch-eobj} and Algorithm~\ref{alg:switch-cumdiff}.}
    \begin{tabular}{lrrrr}
    \hline
          & \multicolumn{4}{c}{(Objective value, TV-norm value)} \\
          \hline
          & \citep{Fei2023binarycontrolpulse}  & Alg.~\ref{alg:overall}w/\citep{Fei2023binarycontrolpulse} & Alg.~\ref{alg:overall}w/\ref{alg:switch-eobj} & Alg.~\ref{alg:overall}w/\ref{alg:switch-cumdiff} \\
          \hline
    Energy2 & (2.930E$-$03, 10) & (1.255E$-$14, 4) & (1.255E$-$14, 4) & (1.255E$-$14, 4) \\
    \hline
    Energy4 & (0.1986, 6) & (0.1611, 6) & (0.1569, 9.2) & (0.1568, 8.4) \\
    \hline
    Energy6 & (0.3046, 10.8) & (0.2049, 10.8) & (0.0526, 20.4) & (0.0632, 19.6) \\
    \hline
    CNOT5 & (0.1968, 10) & (0.1807, 10) & (0.1763, 11) & (0.1792, 9) \\
    \hline
    CNOT10 & (9.431E$-$03, 16) & (5.071E$-$03, 16) & (5.326E$-$07, 28) & (2.163E$-$07, 26) \\
    \hline
    CNOT20 & (8.224E$-$04, 49) & (3.124E$-$08, 49) & (1.728E$-$07, 41) & (3.136E$-$07, 39) \\
    \hline
    NOT2  & (0.1636, 1) & (0.1632, 1) & (0.1632, 1) & (0.1632, 1) \\
    \hline
    NOT6  & (1.270E$-$02, 11) & (6.885E$-$06, 11) & (1.122E$-$06, 0) & (3.267E$-$06, 14) \\
    \hline
    NOT10 & (9.087E$-$04, 20) & (5.132E$-$08, 20) & (2.439E$-$08, 11) & (4.692E$-$08, 17) \\
    \hline
    CircuitH2 & (7.777E$-$02, 8) & (2.747E$-$03, 8) & (1.208E$-$06, 16) & (6.612E$-$05, 36) \\
     \hline
    CircuitLiH & (8.110E$-$01, 18) & (4.054E$-$02, 18) & (1.702E$-$03, 32) & (1.986E$-$03, 20) \\
    \hline
    CircuitBeH2 & (5.115E$-$02, 24) & (1.250E$-$03, 24) & (1.250E$-$03, 28) & (1.250E$-$03, 14) \\
    \hline
    \end{tabular}%

    %    CNOT5 & 0.1968 & 0.1807 & 0.1763 & 0.1792 & 10    & 10    & 11    & 9 \\
%    \hline
%    CNOT10 & 9.431E$-$03 & 5.071E$-$03 & 6.337E$-$07 & 8.584E$-$07 & 16    & 16    & 28    & 24 \\
%    \hline
%    CNOT20 & 8.224E$-$04 & 3.124E$-$08 & 3.481E$-$07 & 7.347E$-$07 & 49    & 49    & 41    & 30 \\
%    \hline
%    NOT2  & 0.1636 & 0.1632 & 0.1632 & 0.1632 & 1     & 1     & 1     & 1 \\
%    \hline
%    NOT6  & 1.270E$-$02 & 6.885E$-$06 & 7.124E$-$06 & 3.146E$-$06 & 11    & 11    & 10    & 13 \\
%    \hline
%    NOT10 & 9.087E$-$04 & 5.132E$-$08 & 4.867E$-$08 & 9.452E$-$08 & 20    & 20    & 16    & 15 \\
    % \end{adjustbox}
  \label{tab:res-obj}%
\end{table}%
% We present the trade-off figure between objective value and TV-norm in Figure~\ref{fig:all}.

% We present how CPU times and numbers of iterations in the switching time optimization vary among problem sizes computed by $2^q\cdot T\cdot N$ in Figure~\ref{fig:time}. We show that CPU times grow exponentially with $q$ because the dimensions of Hamiltonian matrices grow exponentially, while the number of iterations solving the switching time optimization model~\eqref{eq:model-ts} is stable. 
% \begin{figure}
%     \centering
%     \includegraphics[width=0.95\textwidth]{images/time_and_iteration_st.png}
%     \caption{CPU times and numbers of iterations varying among problem sizes computed by $w^q\cdot T\cdot N$. We take the common logarithm of all the values.}
%     \label{fig:time}
% \end{figure}